%% file: RamT21_arxiv_rev.tex
\documentclass[draftcls,onecolumn]{IEEEtran}

\usepackage{cite}
\usepackage{amsmath,amssymb,amsfonts}
\usepackage{dsfont}
\usepackage{algorithmic}
\usepackage{graphicx}
\usepackage{textcomp}
\usepackage{xcolor}
\usepackage{enumitem}
\usepackage{booktabs,pifont} 

\usepackage{mathtools}

\usepackage{amsthm}
\usepackage{arydshln}
\usepackage{multirow}
\usepackage{calc}
\usepackage{blkarray}
\usepackage{url}
\theoremstyle{definition}
\newtheorem{theorem}{Theorem}

\newtheorem{lemma}{Lemma}
\newtheorem{claim}{Claim}

\newtheorem{example}{Example}
\newtheorem{remark}{Remark}
\newtheorem{definition}{Definition}
\usepackage{graphicx,cite}
\usepackage{dblfloatfix}
\usepackage[justification=centering]{caption}
\usepackage{blindtext, graphicx, amsfonts,
	amssymb,multirow,epstopdf}
\usepackage{algorithm}
\usepackage{algorithmic}

\usepackage{hyperref}
\usepackage{comment}

\newcommand{\bfG}{\mathbf{G}}
\newcommand{\bfR}{\mathbf{R}}
\newcommand{\bfP}{\mathbf{P}}

\newcommand{\bfI}{\mathbf{I}}
\newcommand{\bfm}{\mathbf{m}}
\newcommand{\bfs}{\mathbf{s}}
\newcommand{\bft}{\mathbf{t}}
\newcommand{\bfc}{\mathbf{c}}
\newcommand{\bfQ}{\mathbf{Q}}

\newcommand{\bfD}{\mathbf{D}}
\newcommand{\bfA}{\mathbf{A}}
\newcommand{\bfB}{\mathbf{B}}
\newcommand{\bfC}{\mathbf{C}}
\newcommand{\bfx}{\mathbf{x}}
\newcommand{\bfy}{\mathbf{y}}
\newcommand{\bfz}{\mathbf{z}}
\newcommand{\bfW}{\mathbf{W}}
\newcommand{\bfM}{\mathbf{M}}

\newcommand{\bfZ}{\mathbf{Z}}

\newcommand{\bfV}{\mathbf{V}}

\newcommand{\bfX}{\mathbf{X}}

\newcommand{\bfg}{\mathbf{g}}
\newcommand{\calF}{\mathcal{F}}

\newcommand{\calM}{\mathcal{M}}
\newcommand{\calI}{\mathcal{I}}

\title{Numerically stable coded matrix computations via circulant and rotation matrix embeddings}
\author{
	\IEEEauthorblockN{Aditya Ramamoorthy and Li Tang} \\
	\IEEEauthorblockA{Department of Electrical and Computer Engineering \\
	Iowa State University, Ames, IA 50011 \\
	\texttt{\{adityar,litang\}@iastate.edu}
    }
\thanks{
This work was supported in part by the National Science
Foundation (NSF) under Grant CCF-1718470 and Grant CCF-1910840. The material in this work will appear in part at the 2021 IEEE International Symposium on Information Theory.}
}

\date{}
\begin{document}

\maketitle

\begin{abstract}
	Polynomial based methods have recently been used in several works for mitigating the effect of stragglers (slow or failed nodes) in distributed matrix computations. For a system with $n$ worker nodes where $s$ can be stragglers, these approaches allow for an optimal recovery threshold, whereby the intended result can be decoded as long as any $(n-s)$ worker nodes complete their tasks. However, they suffer from serious numerical issues owing to the condition number of the corresponding real Vandermonde-structured recovery matrices; this condition number grows exponentially in $n$. We present a novel approach that leverages the properties of circulant permutation matrices and rotation matrices for coded matrix computation. In addition to having an optimal recovery threshold, we demonstrate an upper bound on the worst-case condition number of our recovery matrices which grows as $\approx O(n^{s+5.5})$; in the practical scenario where $s$ is a constant, this grows polynomially in $n$. Our schemes leverage the well-behaved conditioning of complex Vandermonde matrices with parameters on the complex unit circle, while still working with computation over the reals. Exhaustive experimental results demonstrate that our proposed method has condition numbers that are orders of magnitude lower than prior work.
\end{abstract}

\section{Introduction}
Present day computing needs necessitate the usage of large computation clusters that regularly process huge amounts of data on a regular basis. In several of the relevant application domains such as machine learning, datasets are often so large that they cannot even be stored in the disk of a single server. Thus, both storage and computational speed limitations require the computation to be spread over several worker nodes. Such large scale clusters also present attendant operational challenges. These clusters (which can be heterogeneous in nature) suffer from the problem of ``stragglers'', which are defined as slow nodes (node failures are an extreme form of a straggler). The overall speed of a computational job on these clusters is typically dominated by stragglers in the absence of a sophisticated assignment of tasks to the worker nodes. In particular, simply creating multiple copies of a task to protect against worker node failure can be rather wasteful of computational resources.

In recent years, approaches based on coding theory (referred to as ``coded computation") have been effectively used for straggler mitigation. Coded computation offers significant benefits for specific classes of problems such as matrix computations. The essential idea is to create redundant tasks so that the desired result can be recovered as long as a certain number of worker nodes complete their tasks. For instance, suppose that a designated master node wants to compute $\bfA^T \bfx$ where the matrix $\bfA$ is very large. It can decompose $\bfA$ into block-columns so that $\bfA = [\bfA_0 ~ \bfA_1]$ and assign three worker nodes the tasks of determining $\bfA^T_0 \bfx$, $\bfA^T_1 \bfx$ and $(\bfA^T_0 + \bfA^T_1) \bfx$ respectively. It is easy to see that even if one worker node fails, there is enough information for the master node to compute the final result \cite{lee2018speeding}. Thus, the core idea is to introduce redundancy within the distributed computation by coding across submatrices of the input matrices $\bfA$ and $\bfB$. The worker nodes are assigned computational tasks, such that the master node can decode $\bfA^T \bfB$ as long as a certain minimum number of the worker nodes complete their tasks.

There have been several works, that have exploited the correspondence of coded computation with erasure codes (see \cite{ramamoorthyDTMag20} for a tutorial introduction and relevant references). The matrix computation is embedded into the structure of an underlying erasure code and stragglers are treated as erasures. A scheme is said to have a threshold $\tau$ if the master node can decode the intended result (matrix-vector or matrix-matrix multiplication) as long any $\tau$ nodes complete their tasks. The work of \cite{yu2018straggler,dutta2019optimal} has investigated the tradeoff between the threshold and the tasks assigned to the worker nodes. We discuss related work in more detail in the upcoming Section \ref{sec:background}.

In this work we examine coded computation from the perspective of numerical stability. Erasure coding typically works with operations over finite fields. Solving a linear system of equation over a finite field only requires the corresponding system to be full-rank. However, when operating over the real field, a numerically robust solution can only be obtained if the condition number (ratio of maximum to minimum singular value) \cite{Higham} of the system of the equations is small. It turns out that several of the well-known coded computation schemes that work by polynomial evaluation/interpolation have serious numerical stability issues owing to the high condition number of corresponding real Vandermonde system of equations.  In this work, we present a scheme that leverages the proporties of structured matrices such as circulant permutation matrices and rotation matrices for coded computation. These matrices have eigenvalues that lie on the complex unit circle. Our scheme allows us to exploit the significantly better behaved conditioning of complex Vandermonde matrices while still working with computation over the reals. We also present exhaustive comparisons with existing work.

This paper is organized as follow. Section \ref{sec:prob_form} presents the problem formulation and Section \ref{sec:background} overviews related work and summarizes our contributions. Section \ref{sec:struc_matrices} and \ref{sec:gDMM} discuss our proposed schemes, while Section \ref{sec:comps} presents numerical experiments and comparisons with existing approaches. Section \ref{sec:conclusions} concludes the paper with a discussion of future work. Several of our proofs appear in the Appendix.
\section{Problem Formulation}
\label{sec:prob_form}
Consider a scenario where the master node has a large $t \times r$ matrix $\bfA \in \mathbb{R}^{t \times r}$ and either a $t \times 1$ vector $\bfx \in \mathbb{R}^{t \times 1}$ or a $t \times w$ matrix $\bfB \in \mathbb{R}^{t \times w}$. The master node wishes to compute $\bfA^T \bfx$ or $\bfA^T \bfB$ in a distributed manner over $n$ worker nodes in the matrix-vector and matrix-matrix setting respectively. Towards this end, the master node partitions $\bfA$ (respectively $\bfB$) into $\Delta_A$ (respectively $\Delta_B$) block-columns. Each worker node is assigned $\delta_A \leq \Delta_A$ and $\delta_B \leq \Delta_B$ linearly encoded block-columns of $\bfA_0, \dots, \bfA_{\Delta_A - 1}$ and $\bfB_0, \dots, \bfB_{\Delta_B - 1}$, so that $\delta_A/\Delta_A \leq \gamma_A$ and $\delta_B/\Delta_B \leq \gamma_B$, where $\gamma_A$ and $\gamma_B$ represent the storage fraction constraints for $\bfA$ and $\bfB$ respectively.

In the matrix-vector case, the $i$-th worker is assigned encoded submatrices of $\bfA$ and the vector $\bfx$ and computes their inner product. In the matrix-matrix case it computes pairwise products of submatrices assigned to it (either all or some subset thereof). We say that a given scheme has {\it computation threshold} $\tau$ if the master node can decode the intended result as long as {\it any} $\tau$ out of $n$ worker nodes complete their tasks. In this case we say that the scheme is resilient to $s = n -\tau$ stragglers. We say that this threshold is {\it optimal} if the value of $\tau$ is the smallest possible for the given storage capacity constraints.

The overall goal is to (i) design schemes that are resilient to $s$ stragglers ($s$ is a design parameter), while ensuring that the (ii) desired result can be decoded in a efficient manner, and (iii) the decoded result is numerically robust even in the presence of round-off errors and other sources of noise.

An analysis of numerical stability is closely related to the condition number of matrices.
Let $||\bfM||$ denote the maximum singular value of a matrix $\bfM$ of dimension $l \times l$.
\begin{definition} {\it Condition number.} The condition number of a $l \times l$ matrix $\bfM$ is defined as $\kappa(\bfM) = ||\bfM|| ||\bfM^{-1}||$. It is infinite if the minimum singular value of $\bfM$ is zero.
\end{definition}
Consider the system of equations $\bfM \bfy = \bfz$, where $\bfz$ is known and $\bfy$ is to be determined. If $\kappa(\bfM) \approx 10^{b}$, then the decoded result loses approximately $b$ digits of precision \cite{Higham}. In particular, matrices that are ill-conditioned lead to significant numerical problems when solving linear equations.

\section{Background, Related Work and Summary of Contributions}
\label{sec:background}

A significant amount of prior work \cite{yu2017polynomial,yu2018straggler,duttaCG16,dutta2019optimal} has demonstrated interesting and elegant approaches based on embedding the distributed matrix computation into the structure of polynomials. Specifically, the encoding at the master node can be viewed as evaluating certain polynomials at distinct real values. Each worker node gets a particular evaluation. When at least $\tau$ workers finish their tasks, the master node can decode the intended result by performing polynomial interpolation. The work of \cite{yu2017polynomial} demonstrates that when $\bfA$ and $\bfB$ are split column-wise and $\delta_A = \delta_B = 1$, the optimal threshold for matrix multiplication is $\Delta_A \Delta_B$ and that polynomial based approaches (henceforth referred to as polynomial codes) achieve this threshold. Prior work has also considered other ways in which the matrices $\bfA$ and $\bfB$ can be partitioned. For instance, they can be partitioned both along rows and columns. The work of \cite{yu2018straggler,dutta2019optimal} has obtained threshold results in those cases as well. The so called Entangled Polynomial and Mat-Dot codes \cite{yu2018straggler,dutta2019optimal}, also use polynomial encodings. The key point is that in all these approaches, polynomial interpolation is required when decoding the required result. We note here that to our best knowledge, the idea of embedding matrix multiplication using polynomial maps goes
back much further to Yagle \cite{yagle1995fast} (the motivation there was fast
matrix multiplication).

Polynomial interpolation corresponds to solving a real {\it Vandermonde} system of equations at the master node. In the work of \cite{yu2017polynomial}, this would require solving a $\Delta_A \Delta_B \times \Delta_A \Delta_B$ Vandermonde system. Unfortunately, it can be shown that the condition number of these matrices grows exponentially in $\Delta_A \Delta_B$ \cite{Pan16}. This is a significant drawback and even for systems with around $\Delta_A \Delta_B \approx 30$, the condition number is so high that the decoded results are essentially useless (see Section \ref{sec:comps}).

In Section VII of \cite{yu2018straggler}, it is remarked that when operating over infinite fields such as the reals, one can embed the computation into finite fields to avoid numerical errors. They advocate encoding and decoding over a large enough finite field of prime order $p$. However, this method would require ``quantizing'' real matrices $\bfA$ and $\bfB$ so that the entries are integers. We demonstrate that the performance of this method can be catastrophically bad. In particular, for this method to work, the maximum possible absolute value of each entry of the quantized matrices, $\alpha$ should be such that $\alpha^{2} t < p$, since each entry in the result corresponds to the inner product of columns of $\bfA$ and columns of $\bfB$. This “dynamic range constraint (DRC)” means that the error in the computation depends strongly on the actual matrix entries and the value of $t$ is quite limited. If the DRC is violated, the error in the underlying computation can be catastrophic. Even if the DRC is not violated, the dependence of the error on the entries can make it very bad. We discuss this issue in detail in Section \ref{sec:comps}.

The issue of numerical stability in the coded computation context has been considered in a few recent works \cite{tang2018bound, RamamoorthyTV19,DasTR18,DasR19,DasRV19,FahimC19,subramaniam2019random, das2020coded}. The work of \cite{RamamoorthyTV19,DasR19} presented strategies for distributed matrix-vector multiplication and demonstrated some schemes that empirically have better numerical performance than polynomial based schemes for some values of $n$ and $s$. However, both these approaches work only for the matrix-vector problem. Reference \cite{DasRV19} presents a random convolutional coding approach that applies for both the matrix-vector and the matrix-matrix multiplications problems. Their work demonstrates a computable upper bound on the worst-case condition number of the decoding matrices by drawing on connections with the asymptotic analysis of large Toeplitz matrices. The recent preprint \cite{subramaniam2019random} presents constructions that are based on random linear coding ideas where the encoding coefficients are chosen at random from a continuous distribution. These exhibit better condition number properties.

Reference \cite{FahimC19} which considers an alternative approach for polynomial based schemes by working within the basis of orthogonal polynomials is most closely related to our work. It demonstrates an upper bound on the worst-case condition number of the decoding matrices which grows as $O(n^{2s})$ where $s$ is the number of stragglers that the scheme is resilient to. They also demonstrate experimentally that their performance is better than the polynomial code approach. In contrast we demonstrate an upper bound that is $\approx O(n^{s+5.5})$. Furthermore, in Section \ref{sec:comps} we show that in numerical experiments our worst-case condition numbers are much better than \cite{FahimC19} (even when $s \leq 6$). 

\newcommand{\cmark}{\ding{51}}%
\newcommand{\xmark}{\ding{55}}

\begin{table*}[t]
\caption{{\small Comparison with existing schemes in the literature. The last column indicates the known analytical results about the worst-case condition number of the corresponding recovery matrices. The abbreviations M-V and M-M in the last four rows refer to matrix-vector and matrix-matrix multiplication, respectively. For the M-V cases only the storage fraction $\gamma_A$ is relevant. For the circulant embedding $\tilde{q}$ needs to be prime. The constant $c_1 = 5.5$.}} 
\label{table:comps_summary}

\begin{center}
\begin{small}
\begin{sc}
\begin{tabular}{|p{2.5cm}|p{2cm}|p{3cm}|p{3.5cm}|p{3.5cm}|}
\hline
Code & Storage Fraction $(\gamma_A,\gamma_B)$ & Matrix Split & Threshold ($\tau$) & Condition number\\ \hline
Polynomial \cite{yu2017polynomial} & $1/k_A, 1/k_B$  & Column-wise & $k_A k_B$ & $\geq \Omega(e^\tau)$ \\ \hline
Ent. Polynomial \cite{yu2018straggler} & $1/pk_A, 1/pk_B$  & Row and Column-wise & $p k_A k_B + p - 1$ &  $\geq \Omega(e^\tau)$ \\ \hline
Ortho-Poly\cite{FahimC19} & $1/k_A, 1/k_B$ & Column-wise & $k_A k_B$ & $\leq O(n^{2(n-\tau)})$ \\ \hline
Ortho-Poly \cite{FahimC19} & $1/pk_A, 1/pk_B$ & Row and Column-wise &  $4k_Ak_Bp-2(k_Ak_B+pk_A+pk_B)+k_A+k_B+2p-1$ & $\leq O(n^{2(n-\tau)})$ \\ \hline
Convol. \cite{DasRV19} & $1/k_A, 1/k_B$ & Column-wise & $k_A k_B$ & Computable upper bound \\ \hline
RKRP \cite{subramaniam2019random} & $1/k_A, 1/k_B$ & Column-wise & $k_A k_B$ & Analytical upper bound unknown  \\ \hline
Rot. Embed. (M-V) & $1/k_A$ & Column-wise & $k_A$ & $O(n^{n-\tau + c_1} )$\\ \hline
Circ Embed. (M-V) & $\tilde{q}/k_A(\tilde{q}-1)$ & Column-wise & $k_A$ & $O(n^{n-\tau + c_1})$\\ \hline
Rot. Embed. (M-M)& $1/k_A, 1/k_B$ & Column-wise & $k_A k_B$ & $\leq O(n^{n-\tau + c_1})$ \\ \hline
Rot. Embed. (M-M) & $1/pk_A, 1/pk_B$ & Row and Column-wise & $2pk_A k_B - 1$ & $\leq O(n^{n-\tau + c_1})$ \\ \hline
\bottomrule
\end{tabular}
\end{sc}
\end{small}
\end{center}
\end{table*}%

\subsection{Summary of contributions}
The work of \cite{Pan16} shows that unless all (or almost all) the parameters of the Vandermonde matrix lie on the unit circle, its condition number is badly behaved. However, most of these parameters are complex-valued (except $\pm 1$), whereas our matrices $\bfA$ and $\bfB$ are real-valued. Using complex evaluation points in the polynomial code scheme, will increase the cost of computations approximately four times for matrix-matrix multiplication and around two times for matrix-vector multiplication. This is an unacceptable hit in computation time.

The main idea of our work is to consider alternate embeddings of distributed matrix computations that are based on rotation and circulant permutation matrices. We demonstrate that these are significantly better behaved from a numerical stability perspective. Furthermore, the worker nodes only work with real computation, thus our method does not incur the complex arithmetic overhead.

\begin{itemize}
	\item Our main finding in this paper is that we can work with matrix embeddings that allow the worker nodes to perform real-valued computation. Our scheme (i) continues to have the {\it optimal} threshold of polynomial based approaches when the storage fractions are $\frac{1}{k_A}$ and $\frac{1}{k_B}$ and (ii) enjoys the low condition number of complex Vandermonde matrices with all parameters on the unit circle. In particular, we demonstrate that rotation matrices and circulant permutation matrices of appropriate sizes can be used within the framework of polynomial codes. At the top level, instead of evaluating polynomials at real values, our approach evaluates the polynomials at matrices.
	\item Using these embeddings we show that the worst-case condition number over all $\binom{n}{n-s}$ possible recovery matrices is upper bounded by $\approx O(n^{s+5.5})$. Furthermore, our experimental results indicate that the actual values are significantly smaller, i.e., the analytical upper bounds are pessimistic.
    \item An exhaustive numerical comparison with other approaches in the literature shows that the numerical stability of our scheme is currently the best known.
\end{itemize}
Table \ref{table:comps_summary} contains a comparison of our work with other schemes in the literature. The columns indicate the corresponding storage fractions, matrix splitting methods, threshold and bounds on the condition number.

\section{Numerically Stable Distributed Matrix Computation Schemes}
\label{sec:struc_matrices}
Our schemes in this work will be defined by the encoding matrices used by the master node, which are such that the master node only needs to perform scalar multiplications and additions. The computationally intensive tasks, i.e.,  matrix operations are performed by the worker nodes. We begin by defining certain classes of matrices, discuss their relevant properties and present an example that outlines the basic idea of our work.

In what follows, we let $\rm{i} = \sqrt{-1}$ and let $[m]$ denote the set $\{0, \dots, m-1\}$. For a matrix $\bfM$, $\bfM(i,j)$ denotes its $(i,j)$-th entry, whereas $\bfM_{i,j}$ denotes the $(i,j)$-th block sub-matrix of $\bfM$. We use MATLAB inspired notation at certain places. For instance, $\text{diag}(a_1, a_2, \dots, a_m)$ denotes a $m \times m$ diagonal matrix with $a_i$'s on the diagonal and $\bfM(:,j)$ denotes the $j$-th column of matrix $\bfM$. The notation $\bfM_1 \otimes \bfM_2$ denotes the Kronecker product of $\bfM_1$ and $\bfM_2$ and the superscript $*$ for a matrix denotes the complex conjugation operator.

\begin{definition}{\it Rotation matrix.} The $2 \times 2$ matrix $\bfR_\theta$ below is called a rotation matrix.
	\begin{align}
		\bfR_\theta &= \begin{bmatrix}
			\cos \theta &-\sin \theta\\
			\sin \theta &\cos \theta
		\end{bmatrix}
		= \bfQ \Lambda \bfQ^*, \text{~where} \label{eq:rotmat_eig}\\
		\bfQ &= \frac{1}{\sqrt{2}}\begin{bmatrix}
			\rm{i} &-\rm{i}\\
			1 & 1
		\end{bmatrix}, \text{~and~}
		\Lambda = \begin{bmatrix}
			e^{\rm{i} \theta} & 0\\
			0 & e^{-\rm{i} \theta}
		\end{bmatrix}. \label{eq:Lambda_spec}
	\end{align}
\end{definition}

\begin{definition} {\it Circulant Permutation Matrix.}
\label{defn:circ_perm}
Let $\mathbf{e}$ be a row vector of length $m$ with $\mathbf{e} = [0 ~ 1 ~ 0 ~ \dots ~ 0]$. Let $\bfP$ be a $m \times m$ matrix with $\mathbf{e}$ as its first row. The remaining rows are obtained by cyclicly shifting the first row with the shift index equal to the row index. Then $\bfP^i, i \in [m]$ are said to be circulant permutation matrices. Let $\bfW$ denote the $m$-point Discrete Fourier Transform (DFT) matrix, i.e., $\bfW(i,j)=\frac{1}{\sqrt{m}}\omega_m^{ij}$ for $i \in [m], j \in [m]$ where $\omega_m = e^{\rm{i} \frac{2\pi}{m}}$ denotes the $m$-th root of unity. Then, it can be shown \cite{GrayTCM06} that $\bfP = \bfW \text{diag}(1, \omega_m, \omega_m^{2}, \dots, \omega_m^{(m-1)}) \bfW^{*}$.
\end{definition}

\begin{example}
\label{egs:circ_perm}
	For $m = 4$, the four possible circulation permutation matrices are
	\begin{align*}
	\bfP &= \begin{bmatrix}
	0&1&0&0\\
	0&0&1&0\\
	0&0&0&1\\
	1&0&0&0
	\end{bmatrix},
	\bfP^0=\bfI_4=\begin{bmatrix}
	1&0&0&0\\
	0&1&0&0\\
	0&0&1&0\\
	0&0&0&1
	\end{bmatrix},\\
	\bfP^2 &= \begin{bmatrix}
	0&0&1&0\\
	0&0&0&1\\
	1&0&0&0\\
	0&1&0&0
	\end{bmatrix},
	\bfP^3=\begin{bmatrix}
	0&0&0&1\\
	1&0&0&0\\
	0&1&0&0\\
	0&0&1&0
	\end{bmatrix}.
	\end{align*}
\end{example}

\begin{remark}
Rotation matrices and circulant permutation matrices have the useful property that they are ``real'' matrices with complex eigenvalues that lie on the unit circle. We use this property extensively in the sequel.
\end{remark}

\begin{definition}{\it Vandermonde Matrix.}
	A $m \times m$ Vandermonde matrix $\bfV$ with parameters $s_0, s_1, \dots, s_{m-1} \in \mathbb{C}$ is such that $\bfV(i,j) = s_{j}^{i}, i \in [m], j \in [m]$.
	If the $s_i$'s are distinct, then $\bfV$ is nonsingular \cite{hornJ91}. In this work, we will also assume that the $s_i$'s are non-zero.
\end{definition}

\noindent {\bf Condition Number of Vandermonde Matrices:} Let $\bfV$ be a $m \times m$ Vandermonde matrix with parameters $s_0, s_1, \dots, s_{m-1}$. The following facts about $\kappa(\bfV)$ follow from prior work \cite{Pan16}.

\begin{itemize}[wide, labelwidth=!, labelindent=0pt]
	\item {\it Real Vandermonde matrices.} If $s_i \in \mathbb{R}, i \in [m]$, i.e., if $\bfV$ is a real Vandermonde matrix, then it is known that its condition number is exponential in $m$.
	\item {\it Complex Vandermonde matrices with parameters ``not'' on the unit circle.} Suppose that the $s_i$'s are complex and let $s_{+} = \max_{i=0}^{m-1} |s_i|$. If $s_{+} > 1$ then $\kappa(\bfV)$ is exponential in $m$. Furthermore, if $1/|s_i| \geq \nu > 1$ for at least $\beta \leq m$ of the $m$ parameters, then $\kappa(\bfV)$ is exponential in $\beta$.
\end{itemize}
Based on the above facts, the only scenario where the condition number is somewhat well-behaved is if most or all of the parameters of $\bfV$ are complex and lie on the unit-circle.
In the Appendix \ref{sec:proof_vand_cond_no}, we show the following result which is one of our key technical contributions.
\begin{theorem}
	\label{thm:cond_no_vand}
	Consider a $m \times m$ Vandermonde matrix $\bfV$ where $m < q$ (where $q$ is odd) with distinct parameters $\{s_0, s_1, \dots, s_{m-1}\} \subset \{1, \omega_q, \omega_q^2, \dots, \omega_q^{q-1}\}$. Let $c_1 = 5.5$. Then,
	\begin{align*}
		\kappa(\bfV) \leq O(q^{q - m + c_1}).
	\end{align*}
\end{theorem}
\begin{remark}
For the remainder of the paper, we continue to use this theorem with $c_1 = 5.5$. If $q-m$ is a constant, then $\kappa(\bfV)$ grows only polynomially in $q$. In the subsequent discussion, we will leverage Theorem \ref{thm:cond_no_vand} extensively.
\end{remark}

\begin{example}{ \it Polynomial Codes.}
\label{eg:poly_codes}
Consider the matrix-vector case where $\Delta_A =3$ and $\delta_A = 1$. In the polynomial approach, the master node forms
$\bfA(z) = \bfA_0 + \bfA_1 z + \bfA_2 z^2$ and evaluates it at distinct real values $z_1, \dots, z_n$. The $i$-th evaluation is sent to the $i$-th worker node which computes $\bfA^T(z_i) \bfx$. From polynomial interpolation, it follows that as long as the master node receives results from any three workers, it can decode $\bfA^T \bfx$. However, when $\Delta_A$ is large, the interpolation is numerically unstable \cite{Pan16}.
\end{example}

The basic idea of our approach to tackle the numerical stability issue is as follows. We further split each $\bfA_i$ into two equal sized block-columns. Thus, we now have six block-columns, indexed as $\bfA_0, \dots \bfA_5$. Consider the $6 \times 2$ matrix defined below; its columns are specified by $\bfg_0$ and $\bfg_1$.
\begin{align*}
[\bfg_{0} ~ \bfg_{1}] = &
\begin{bmatrix}
\bfI\\
\bfR_\theta^i\\
\bfR_\theta^{2i}
\end{bmatrix}
\end{align*}
The master node forms ``two'' encoded matrices for the $i$-th worker: $\sum_{j=0}^5 \bfA_j \bfg_{0}(j)$ and $\sum_{j=0}^5 \bfA_j \bfg_{1}(j)$ (where $\bfg_i(l)$ denotes the $l$-th component of the vector $\bfg_i$). Thus, the storage capacity constraint fraction $\gamma_A$ is still $\frac{1}{3}$.

Worker node $i$ computes the inner product of these two encoded matrices with $\bfx$ and sends the result to the master node. It turns out that in this case when any three workers $i_0, i_1, $ and $i_2$ complete their tasks, the decodability and numerical stability of recovering $\bfA^T \bfx$ depends on the condition number of the following matrix.
\begin{align*}
\begin{bmatrix}
\bfI & \bfI & \bfI\\
\bfR_\theta^{i_0} & \bfR_\theta^{i_1} & \bfR_\theta^{i_2}\\
\bfR_\theta^{2i_0} & \bfR_\theta^{2i_1} & \bfR_\theta^{2i_2}
\end{bmatrix}.
\end{align*}
Using the eigen-decomposition of $\bfR_\theta$ ({\it cf.} (\ref{eq:rotmat_eig})) the above block matrix can expressed as
	\begin{align*}
		\begin{bmatrix}
			\bfQ&0 & 0 \\
			0 &\bfQ& 0\\
			0 & 0 &\bfQ
		\end{bmatrix}
		\underbrace{\begin{bmatrix}
			\bfI & \bfI & \bfI\\
			\Lambda^{i_0} &  \Lambda^{i_1} &\Lambda^{i_2}\\
			\Lambda^{2i_0} & \Lambda^{2i_1} & \Lambda^{2i_{2}}
		\end{bmatrix}}_{\mathbf{\Sigma}}
        \begin{bmatrix}
			\bfQ^*&0 &0 \\
			0 &\bfQ^*& 0\\
			0 & 0 &\bfQ^*
		\end{bmatrix}.
\end{align*}
As the pre- and post-multiplying matrices are unitary, the condition number of the above matrix only depends on the properties of the middle matrix, denoted by $\mathbf{\Sigma}$. In what follows, we show that upon appropriate column and row permutations, $\mathbf{\Sigma}$ can be shown equivalent to a block diagonal matrix where each of the blocks is a Vandermonde matrix with parameters on the unit circle. Thus, the matrix is invertible if the corresponding parameters are distinct. Furthermore, even though we use real computation, the numerical stability of our scheme depends on Vandermonde matrices with parameters on the unit circle. Theorem \ref{thm:cond_no_vand} shows that the condition number of such matrices is much better behaved.

In the sequel we show that this argument can be significantly generalized and adapted for the case of circulant permutation embeddings. The matrix-matrix case requires the development of more ideas that we also present. In this section we consider $(i)$ the matrix-vector case where the storage fraction $\gamma_A = 1/k_A$ and $(ii)$ the matrix-matrix case where the storage fractions are $\gamma_A = 1/k_A, \gamma_B = 1/k_B$ respectively.
\subsection{Matrix Splitting Scheme}

We partition the matrices $\bfA$ and $\bfB$ into $\Delta_A = k_A \ell$ and $\Delta_B = k_B \ell$ block-columns respectively. However, we use two indices to refer to their respective constituent block-columns as this simplifies our later presentation. To avoid confusion, we use the subscript $\langle i,j \rangle$ to refer to the corresponding $(i,j)$-th block-columns. In particular $\bfA_{\langle i,j\rangle}, i \in [k_A], j \in [\ell]$ and $\bfB_{\langle i,j\rangle},i \in [k_B], j \in [\ell]$ refer to the $(i,j)$-th block-column of $\bfA$ and $\bfB$ respectively, such that
\begin{align}
\bfA &= [\bfA_{\langle 0,0 \rangle}~\dots~\bfA_{\langle 0,\ell-1 \rangle} ~|~\dots~|~\bfA_{\langle k_A-1,0 \rangle}~\dots~\bfA_{\langle k_A-1,\ell-1 \rangle}], \text{~and} \nonumber\\
\bfB &= [\bfB_{\langle 0,0 \rangle}~\dots~\bfB_{\langle 0,\ell-1 \rangle} ~|~\dots~|~\bfB_{\langle k_B-1,0 \rangle}~\dots~\bfB_{\langle k_B-1,\ell-1 \rangle}]. \label{eq:subdivide_A}
\end{align}
\subsection{Distributed Matrix-Vector Multiplication}
\label{sec:mat_vec_mul}
In the matrix-vector case, the encoding matrix for $\bfA$ will be specified by a $k_A \ell \times n \ell$ ``generator'' matrix $\bfG$ such that
\begin{align}
	\hat{\bfA}_{\langle i, j \rangle} = \sum_{\alpha \in [k_A], \beta \in [\ell]} \bfG(\alpha \ell + \beta,i \ell + j) \bfA_{\langle \alpha,\beta \rangle} \label{eq:encoding_rule_A}
\end{align}
for $i \in [n], j \in [\ell]$. The worker node $i$ stores $\hat{\bfA}_{\langle i, j \rangle}$ for $j \in [\ell]$ and $\bfx$, i.e., it stores $\gamma_A = \ell/\Delta_A = 1/k_A$ fraction of matrix $\bfA$. Furthermore, it computes  $\hat{\bfA}^T_{\langle i, j \rangle} \bfx$ for $j \in [\ell]$ and transmits them to the master node. 

Thus, the master node receives $\hat{\bfA}^T_{\langle i, j \rangle} \bfx$ of length $r/\Delta_A$ for $j \in [\ell]$ from a certain number of worker nodes and wants to decode $\bfA^T \bfx$ of length $r$. Based on our encoding scheme, this can be done by solving a $\Delta_A \times \Delta_A$ linear system of equations $r/\Delta_A$ times. The structure of this linear system is inherited from the encoding matrix $\bfG$. The precise details of the encoding schemes can be found in Algorithm \ref{Alg:MV_Embedding_Scheme} (an example appears above).

\begin{algorithm}[t]
	\caption{Encoding scheme for distributed matrix-vector multiplication}
	\label{Alg:MV_Embedding_Scheme}
   \textbf{Input:} Matrix $\bfA$ and vector $\bfx$. Storage fraction $\gamma_A = 1/k_A$, positive integer $\ell$ and encoding matrix $\bfG$ of dimension $k_A \ell \times n \ell$.\\
   \textbf{Output:} Worker task assignment.
   \begin{algorithmic}
       \STATE Partition $\bfA$ into $\Delta_A$ block-columns as in (\ref{eq:subdivide_A}).
        \FOR{$i=0$ {\bfseries to} $n-1$}
        \STATE Worker $i$ is assigned $\hat{\bfA}_{\langle i, j \rangle} = \sum_{\alpha \in [k_A], \beta \in [\ell]} \bfG(\alpha \ell + \beta,i \ell + j) \bfA_{\langle \alpha,\beta \rangle}$, for all $j \in [\ell]$ and the vector $\bfx$.
        \ENDFOR
        \STATE Worker $i$ computes $\hat{\bfA}_{\langle i, j \rangle}^T \bfx$ for $j \in [\ell]$.
   \end{algorithmic}
\end{algorithm}

\subsubsection{Rotation Matrix Embedding}
Let $q$ be an odd number such that $q \geq n$, $\theta = 2\pi/q$ and $\ell=2$ ({\it cf.} block column decomposition in (\ref{eq:subdivide_A})). We choose the generator matrix such that its $(i,j)$-th block-submatrix for $i \in [k_A], j \in [n]$ is given by
\begin{align}
  \bfG^{rot}_{i,j} & = \bfR_\theta^{ji}. \label{eq:rot_mat_generator}
\end{align}

\begin{theorem}
\label{thm:rotmat_mv}
	The threshold for the rotation matrix based scheme specified above is $k_A$. Furthermore, the worst-case condition number of the recovery matrices is upper bounded by $O(q^{q-k_A + c_1})$.
\end{theorem}
\begin{proof}
Suppose that workers indexed by $i_0, \dots, i_{k_A-1}$ complete their tasks. We extract the corresponding block-columns of $\bfG^{rot}$ to obtain
	\begin{align*}
		\tilde{\bfG}^{rot} =
		\begin{bmatrix}
			\bfI & \bfI & \cdots & \bfI\\
			\bfR_\theta^{i_0} &  \bfR_\theta^{i_1} &\cdots &\bfR_\theta^{i_{k_A-1}}\\
			\vdots & \vdots &\ddots &\vdots\\
			\bfR_\theta^{i_0(k_A-1)} & \bfR_\theta^{i_1(k_A-1)} & \cdots & \bfR_\theta^{i_{k_A-1}(k_A-1)}
		\end{bmatrix}.
	\end{align*}
We note here that the decoder attempts to recover each entry of $\bfA_{\langle i,j \rangle}^T \bfx$ from the results sent by the worker nodes. Thus, we can equivalently analyze the decoding by considering the system of equations as
	\begin{align*}
		\bfm \tilde{\bfG}^{rot} = \bfc,
	\end{align*}
where $\bfm, \bfc \in \mathbb{R}^{1 \times k_A \ell}$  are row-vectors such that
	\begin{align*}
		\bfm &=[\bfm_{0},\cdots, \bfm_{k_A-1}]\\
		 &= [\bfm_{\langle 0,0 \rangle}, \cdots, \bfm_{\langle 0,\ell-1 \rangle}, \cdots, \cdots, \bfm_{\langle k_A-1, 0 \rangle}, \cdots, \bfm_{\langle k_A-1, \ell-1 \rangle}],\\ \text{~and} \\
		\bfc &=[\bfc_{i_0}, \cdots, \bfc_{i_{k_A-1}}] \\&= [\bfc_{\langle i_0,0 \rangle},\cdots, \bfc_{\langle i_0,\ell-1 \rangle} , \cdots, \cdots, \bfc_{\langle i_{k_A-1},0\rangle}, \cdots, \bfc_{\langle i_{k_A-1}, \ell-1\rangle}].
	\end{align*}
In the expression above, terms of the form $\bfm_{\langle i, j\rangle}$ and $\bfc_{\langle i, j\rangle}$ are scalars.
We need to analyze $\kappa(\tilde{\bfG}^{rot})$. Towards this end, using the eigenvalue decomposition of $\bfR_\theta$, we have
	\begin{align}
		\tilde{\bfG}^{rot} &=
		\begin{bmatrix}
			\bfQ&&\\
			&\ddots&\\
			&&\bfQ
		\end{bmatrix}
		\mathbf{\tilde{\Lambda}} \begin{bmatrix}
			\bfQ^*&&\\
			&\ddots&\\
			&&\bfQ^*
		\end{bmatrix}, \text{~where} \label{eq:eig_decomp_rot}\\
		\mathbf{\tilde{\Lambda}} &= \begin{bmatrix}
			\bfI & \bfI & \cdots & \bfI\\
			\Lambda^{i_0} &  \Lambda^{i_1} &\cdots &\Lambda^{i_{k_A-1}}\\
			\vdots & \vdots &\ddots &\vdots\\
			\Lambda^{i_0(k_A-1)} & \Lambda^{i_1(k_A-1)} & \cdots & \Lambda^{i_{k_A-1}(k_A-1)}
		\end{bmatrix} \nonumber
	\end{align}
	and $\Lambda$ is specified in (\ref{eq:Lambda_spec}).
	Note that the pre- and post-multiplying matrices in the RHS of (\ref{eq:eig_decomp_rot}) above are both unitary. Therefore $\kappa(\tilde{\bfG}^{rot})$ is the same as  $\kappa(\mathbf{\tilde{\Lambda}})$ \cite{hornJ91}.
	
	Using Claim \ref{claim:diag_block_matrix} in the Appendix \ref{sec:aux_claims}, we can permute $\mathbf{\tilde{\Lambda}}$ to put it in block-diagonal form so that
	\begin{align*}
		\mathbf{\tilde{\Lambda}}_d = \begin{bmatrix}
			\mathbf{\tilde{\Lambda}}_d[0] & \mathbf{0}\\
			\mathbf{0} & \mathbf{\tilde{\Lambda}}_d[1]
		\end{bmatrix},
	\end{align*}
	where $\mathbf{\tilde{\Lambda}}_d[0]$ and $\mathbf{\tilde{\Lambda}}_d[1]$ are Vandermonde matrices with parameter sets $\{e^{{\rm{i}}\theta i_0}, \dots, e^{{\rm{i}}\theta i_{k_A-1}}\}$ and $\{e^{-{\rm{i}}\theta i_0}, \dots, e^{-{\rm{i}}\theta i_{k_A-1}}\}$ respectively. Note that these parameters are distinct points on the unit circle. Thus, $\mathbf{\tilde{\Lambda}}_d[0]$ and $\mathbf{\tilde{\Lambda}}_d[1]$ are both invertible which implies that $\tilde{\Lambda}$ is invertible. This allows us to conclude that the threshold of the scheme is $k_A$. The upper bound on the condition number follows from Theorem \ref{thm:cond_no_vand}.
\end{proof}

\begin{algorithm}[t]
	\caption{Decoding Algorithm for Circulant Permutation Scheme}
	\label{Alg:MatVecMul}
	\textbf{Input:} $\bfG^{circ}_\calI$ where $|\calI| = k_A$ (block-columns of $G$ corresponding to block-columns in $\calI$). Row vector $\bfc$ corresponding to observed values in one system of equations. Permutation $\pi$ specified in the proof of Theorem \ref{thm:circ_perm_scheme}. \\
    \textbf{Output:} $\bfm$ which is the solution to $\bfm \bfG^{circ}_\calI = \bfc$.
	\begin{algorithmic}
		\STATE {\bfseries 1. procedure:} Block Fourier Transform and permute $\bfc$.
		\FOR{$j=0$ {\bfseries to} $k_A-1$}
		\STATE Apply FFT to $\bfc_{i_j}=[\bfc_{\langle i_j,0 \rangle}, \cdots, \bfc_{\langle i_j,\tilde{q}-1 \rangle}]$ to obtain $\bfc_{i_j}^\calF = [\bfc_{\langle i_j,0 \rangle}^\calF, \cdots, \bfc_{\langle i_j,\tilde{q}-1 \rangle}^\calF]$.
		\ENDFOR
		\STATE Permute $\bfc^\calF = [\bfc_{i_0}^\calF, \cdots, \bfc_{i_{k_A-1}}^\calF]$ by $\pi$ to obtain $\bfc^{\calF, \pi} = [\bfc_{0}^{\calF,\pi}, \cdots, \bfc_{{\tilde{q}-1}}^{\calF, \pi}]$ where $\bfc_{j}^{\calF,\pi}=[\bfc_{\langle i_0,j \rangle}^\calF, \bfc_{\langle i_1,j \rangle}^\calF, \cdots, \bfc_{\langle i_{k_A-1},j \rangle}^\calF]$, for $j=0, \dots, \tilde{q}-1$.
		\STATE {\bfseries end procedure}
		\STATE {\bfseries 2. procedure:} Decode $\bfm^{\calF, \pi}$ from $\bfc^{\calF,\pi}$.
		\STATE	For $i \in \{1, \dots, \tilde{q}-1\}$, decode $\bfm_i^{\calF, \pi}$ from  $\bfc_i^{\calF, \pi}$ by polynomial interpolation or matrix inversion of $\tilde{\bfG}^\calF_d[i]$ (see (\ref{eq:diagonal_block_tilde_G_F}) in Appendix \ref{sec:proof_circ_perm}).
		Set $\bfm_0^{\calF, \pi} = [0,\cdots, 0]$.
		\STATE {\bfseries end procedure}
		\STATE {\bfseries 3. procedure:} Inverse permute and Block Inverse Fourier Transform  $\bfm^{\calF, \pi}$.
		\STATE Permute $\bfm^{\calF, \pi}$ by $\pi^{-1}$ to obtain $\bfm^\calF = [\bfm_{0}^\calF, \cdots, \bfm_{{k_A-1}}^\calF]$. Apply inverse FFT to each $\bfm_{i}^\calF$ in $\bfm^\calF$ to obtain $\bfm=[\bfm_{0}, \cdots, \bfm_{{k_A-1}}]$.
		\STATE {\bfseries end procedure}
	\end{algorithmic}
\end{algorithm}

\noindent {\bf Complexity Analysis:} Creating an encoded matrix requires a total of $\Delta_A$ scalar multiplications and $\Delta_A-1$ additions of block-columns of size $t \times r/\Delta_A$. Therefore, the total encoding complexity is given by $O(rtn)$. Each worker node computes the product of submatrix of size $r/\Delta_A \times t$ with a vector of size $t$, i.e., the computational cost is $O(rt/\Delta_A)$. Finally, the decoding process involves inverting a $\Delta_A \times \Delta_A$ matrix once and using the inverse to solve $r/\Delta_A$ systems of equations. Thus, the overall decoding complexity is $O(\Delta_A^3 + r \Delta_A)$ where typically, $r\gg \Delta_A^2$.

\subsubsection{Circulant Permutation Embedding}
Let $\tilde{q}$ be a {\it prime} number which is greater than or equal to $n$. We set $\ell=\tilde{q}-1$, so that $\bfA$ is sub-divided into $k_A(\tilde{q}-1)$ block-columns as in (\ref{eq:subdivide_A}).
In this embedding we have an additional step. Specifically, the master node generates the following ``precoded'' matrices.
\begin{align}
	\bfA_{\langle i,\tilde{q}-1 \rangle}=-\sum_{j=0}^{\tilde{q}-2}\bfA_{\langle i,j \rangle}, i \in [k_A]. \label{eq:precoded_mat}
\end{align}
In the subsequent discussion, we work with the set of block-columns $\bfA_{\langle i,j \rangle}$  for $i \in [k_A], j \in [\tilde{q}]$. The coded submatrices $\hat{\bfA}_{\langle i,j \rangle}$ for $i \in [n], j \in [\tilde{q}]$ are generated by means of a $k_A \tilde{q} \times n \tilde{q}$ matrix $\bfG^{circ}$ using Algorithm \ref{Alg:MV_Embedding_Scheme}.
%
The $(i,j)$-th block of $\bfG^{circ}$ can be expressed as
\begin{align}
	\label{eq:circ_perm_generator}
    \bfG^{circ}_{i,j} = \bfP^{ji}, \text{~for~} i \in [k_A], j \in [n],
\end{align}
where the matrix $\bfP$ denotes the $\tilde{q} \times \tilde{q}$ circulant permutation matrix introduced in Definition \ref{defn:circ_perm}. For this scheme the storage fraction $\gamma_A = \tilde{q}/(k_A(\tilde{q}-1))$, i.e., it is slightly higher than $1/k_A$.

\begin{theorem}
	\label{thm:circ_perm_scheme}
	The threshold for the circulant permutation based scheme specified above is $k_A$. Furthermore, the worst-case condition number of the recovery matrices is upper bounded by $O(\tilde{q}^{\tilde{q}-k_A + c_1})$ and the scheme can be decoded by using Algorithm \ref{Alg:MatVecMul}.
\end{theorem}
The proof appears in the Appendix \ref{sec:proof_circ_perm}. It is conceptually similar to the proof of Theorem \ref{thm:rotmat_mv} and relies critically on the fact that all eigenvalues of $\bfP$ lie on the unit circle and that $\bfP$ can be diagonalized by the DFT matrix $\bfW$.

\noindent {\bf Complexity Analysis:} The complexity analysis closely mirrors the analyses for the case of the rotation matrix embedding. However, we note that for the circulant permutation embedding, the $\hat{\bfA}_{\langle i, j \rangle}$'s can simply be generated by additions since $\bfG^{circ}$ is a binary matrix. Furthermore, the fact that $\bfP$ can be diagonalized by the DFT matrix $\bfW$ suggests an efficient decoding algorithm where the fast Fourier Transform (FFT) plays a key role (see Algorithm \ref{Alg:MatVecMul}). In particular, we have the following claim.
\begin{claim}
\label{claim:fft_decoding_circ}
	The decoding complexity of recovering $\bfA^T \bfx$ is $O(r (\log \tilde{q} + \log^2 k_A))$.
\end{claim}
\begin{remark}
	Both circulant permutation matrices and rotation matrices allow us to achieve a specified threshold for distributed matrix vector multiplication. The required storage fraction $\gamma_A$ is slightly higher for the circulant permutation case and it requires $\tilde{q}$ to be prime. However, it allows for an efficient FFT based decoding algorithm. On the other hand, the rotation matrix case requires a smaller $\Delta_A$, but the decoding requires solving the corresponding system of equations the complexity of which can be cubic in $\Delta_A$. We note that when the size of $\bfA$ is large, the decoding time will be much lesser than the worker node computation time; we demonstrate this numerically as well in Section \ref{sec:comps}. In Section \ref{sec:comps}, we show results that demonstrate that the normalized mean-square error when circulant permutation matrices are used is lower than the rotation matrix case.
\end{remark}

\subsection{Distributed Matrix-Matrix Multiplication}
\label{sec:mat_mat_mul}

The matrix-matrix case requires the introduction of newer ideas within this overall framework. In this case, a given worker obtains encoded block-columns of both $\bfA$ and $\bfB$ and representing the underlying computations is somewhat more involved. Once again we let $\theta = 2\pi/q$, where $q \geq n$ ($n$ is the number of worker nodes) is an odd integer and set $\ell = 2$. Furthermore, let $k_A k_B < n$. The $(i,j)$-th blocks of the encoding matrices are given by appropriate powers of rotation matrices, i.e.,
\begin{align*}
\bfG^A_{i,j} &= \bfR_\theta^{ji}, \text{~for~} i \in [k_A], j \in [n], \text{~and}\\
\bfG^B_{i,j} &= \bfR_\theta^{(j k_A) i}, \text{~for~} i \in [k_B], j \in [n].
\end{align*}
The master node operates according to the encoding rule discussed previously in the matrix-vector case; the details can be found in Algorithm \ref{Alg:MM_Embedding_Scheme}. Thus, each worker node stores $\gamma_A = 1/k_A$ and $\gamma_B = 1/k_B$ fraction of $\bfA$ and $\bfB$ respectively. The $i$-th worker node computes the pair-wise product of the matrices $\hat{\bfA}_{\langle i,l_1 \rangle}^T \hat{\bfB}_{\langle i,l_2 \rangle}$ for $l_1, l_2 = 0,1$ and returns the result to the master node. Thus, the master node needs to recover all pair-wise products of the form $\bfA_{\langle i,\alpha \rangle}^T \bfB_{\langle j,\beta \rangle}$ for $i \in [k_A], j \in [k_B]$ and $\alpha,\beta = 0,1$. Let $\bfZ$ denote a $1 \times 4k_A k_B$ block matrix that contains all of these pair-wise products. The details of the encoding scheme can be found in Algorithm \ref{Alg:MM_Embedding_Scheme} (an example appears below).
\begin{algorithm}[t]
	\caption{Encoding scheme for distributed matrix-matrix multiplication}
	\label{Alg:MM_Embedding_Scheme}
   \textbf{Input:} Matrices $\bfA$ and $\bfB$. Storage fractions $\gamma_A = 1/k_A,\gamma_B = 1/k_B$, positive integer $\ell$ and encoding matrices $\bfG^A$ and $\bfG^B$ of dimensions $k_A \ell \times n \ell$ and $k_B \ell \times n$ respectively.\\
   \textbf{Output:} Worker task assignment.
   \begin{algorithmic}
       \STATE Partition $\bfA$ and $\bfB$ into $\Delta_A$ and $\Delta_B$ block-columns as in (\ref{eq:subdivide_A}).
        \FOR{$i=0$ {\bfseries to} $n-1$}
        \STATE Worker $i$ is assigned
        \begin{align*}
        \hat{\bfA}_{\langle i, j \rangle} &= \sum_{\alpha \in [k_A], \beta \in [\ell]} \bfG^A(\alpha \ell + \beta,i \ell + j) \bfA_{\langle \alpha,\beta \rangle}, \text{~and}\\
        \hat{\bfB}_{\langle i, j \rangle} &= \sum_{\alpha \in [k_B], \beta \in [\ell]} \bfG^B(\alpha \ell + \beta,i \ell + j) \bfB_{\langle \alpha,\beta \rangle}
        \end{align*}
        for all $j \in [\ell]$.
        \ENDFOR
        \STATE Worker $i$ computes $\hat{\bfA}_{\langle i, l_1 \rangle}^T \hat{\bfB}_{\langle i, l_2 \rangle}$ for all pairs $l_1 \in [\ell], l_2 \in [\ell]$.
   \end{algorithmic}
\end{algorithm}

\begin{example}
	Suppose $k_A = 2$, $k_B=2$. Let $n=q = 5$, $\theta=2\pi/5$. The matrix $\bfA$ and $\bfB$ can be partitioned as follows.
	\begin{align*}
	\bfA &= [\bfA_{\langle 0,0 \rangle} ~ ~\bfA_{\langle 0,1 \rangle} ~ |~\bfA_{\langle 1,0 \rangle} ~ ~\bfA_{\langle 1,1 \rangle} ], \text{~and} \\
	\bfB &= [\bfB_{\langle 0,0 \rangle} ~ ~\bfB_{\langle 0,1 \rangle} ~ |~\bfB_{\langle 1,0 \rangle} ~ ~\bfB_{\langle 1,1 \rangle} ].
	\end{align*}
	The encoding matrices $\bfG^A$ and $\bfG^B$ are given by
	\begin{align*}
	\bfG^A  &= \begin{bmatrix}
	\bfI &\bfI& \bfI & \bfI & \bfI\\
	\bfI &\bfR_\theta& \bfR_\theta^2 & \bfR_\theta^3 & \bfR_\theta^4
	\end{bmatrix}, \text{~and}\\
	\bfG^B &= \begin{bmatrix}
	\bfI &\bfI& \bfI & \bfI & \bfI\\
	\bfI &\bfR_\theta^2& \bfR_\theta^4 & \bfR_\theta^6 & \bfR_\theta^8
	\end{bmatrix}.
	\end{align*}
	
	Thus, for the $i$-th worker node, the encoded matrices are obtained as
	\begin{align*}
	\hat{\bfA}_{\langle i,0 \rangle} &= \bfA_{\langle 0, 0\rangle} + \bfR^i_\theta(0, 0)\bfA_{\langle 1, 0\rangle}+  \bfR^i_\theta(1, 0)\bfA_{\langle 1, 1\rangle},\\
		\hat{\bfA}_{\langle i,1 \rangle} &= \bfA_{\langle 0, 1\rangle} + \bfR^i_\theta(0, 1)\bfA_{\langle 1, 0\rangle}+  \bfR^i_\theta(1, 1)\bfA_{\langle 1, 1\rangle},\\
		\hat{\bfB}_{\langle i,0 \rangle} &= \bfB_{\langle 0, 0\rangle} + \bfR^{2i}_\theta(0, 0)\bfB_{\langle 1, 0\rangle}+  \bfR^{2i}_\theta(1, 0)\bfB_{\langle 1, 1\rangle}, \text{~and}\\
	    \hat{\bfB}_{\langle i,1 \rangle} &= \bfB_{\langle 0, 1\rangle} + \bfR^{2i}_\theta(0, 1)\bfB_{\langle 1, 0\rangle}+  \bfR^{2i}_\theta(1, 1)\bfB_{\langle 1, 1\rangle}.
	\end{align*}
%
%
	The $i$-th worker node computes $\hat{\bfA}_{\langle i,0 \rangle}^T \hat{\bfB}_{\langle i,0 \rangle}$,  $\hat{\bfA}_{\langle i,0 \rangle}^T \hat{\bfB}_{\langle i,1 \rangle}$, $\hat{\bfA}_{\langle i,1 \rangle}^T \hat{\bfB}_{\langle i,0 \rangle}$, $\hat{\bfA}_{\langle i,1 \rangle}^T \hat{\bfB}_{\langle i,1 \rangle}$. We can represent the computations in the $i$-th worker node using Kronecker products. We take $\hat{\bfA}_{\langle i,0 \rangle}^T \hat{\bfB}_{\langle i,1 \rangle}$ as an example. Let $\bfZ$ denote a $1 \times 16$ block matrix that contains all of the pair-wise products $\bfA_{\langle a, k_1 \rangle}^T\bfB_{\langle b, k_2 \rangle}$, $a,b,k_1, k_2 = 0, 1$. Consider the following vector (of length 16).

	\begin{align*}
	 \begin{bmatrix}
	\bfI(0, 0) \\ \bfI(1, 0) \\ \bfR_\theta^i(0, 0) \\ \bfR_\theta^i(1, 0)
	\end{bmatrix} \otimes \begin{bmatrix}
	\bfI(0, 1) \\ \bfI(1, 1) \\ \bfR_\theta^{2i}(0, 1) \\ \bfR_\theta^{2i}(1, 1)
	\end{bmatrix}.
	\end{align*}
	Then the computation of $\hat{\bfA}_{\langle i,0 \rangle}^T \hat{\bfB}_{\langle i,1 \rangle}$ can be denoted as the product of each of the elements of $\bfZ$ with the corresponding component of the above vector followed by their sum. For the sake of convenience we represent this operation by the $\cdot$ operator below.
	Then we can verify that the computations in $i$-th worker node can be denoted as
	\begin{align*}
	\bfZ \cdot \begin{bmatrix}
	\bfI \\
	\bfR_\theta^i
	\end{bmatrix} \otimes
	\begin{bmatrix}
	\bfI \\
	\bfR_\theta^{2i}
	\end{bmatrix}
	\end{align*}
	
	Suppose that four different worker nodes $i_0, i_1, i_2, i_3$ have finished their work, the master node obtain
	\begin{align*}
	\bfZ \cdot\bfG_d = \bfZ \cdot  \bigg( \begin{bmatrix}
	\bfI \\
	\bfR_\theta^{i_0}
	\end{bmatrix} \otimes
	\begin{bmatrix}
	\bfI \\
	\bfR_\theta^{2i_0}
	\end{bmatrix} ~\bigg{|}~
	\begin{bmatrix}
	\bfI \\
	\bfR_\theta^{i_1}
	\end{bmatrix} \otimes
	\begin{bmatrix}
	\bfI \\
	\bfR_\theta^{2i_1}
	\end{bmatrix} ~\bigg{|}~
	\begin{bmatrix}
	\bfI \\
	\bfR_\theta^{i_2}
	\end{bmatrix} \otimes
	\begin{bmatrix}
	\bfI \\
	\bfR_\theta^{2i_2}
	\end{bmatrix} ~\bigg{|}~
	\begin{bmatrix}
	\bfI \\
	\bfR_\theta^{i_3}
	\end{bmatrix} \otimes
	\begin{bmatrix}
	\bfI \\
	\bfR_\theta^{2i_3}
	\end{bmatrix}
	\bigg)
	\end{align*}
	
	We formalize the above construction and prove the $\bfG_d$ has full rank in Theorem \ref{thm:matMatRotEmbed}.
	\end{example}
	
\begin{theorem}
	\label{thm:matMatRotEmbed}
	The threshold for the rotation matrix based matrix-matrix multiplication scheme is $k_A k_B$. The worst-case condition number is bounded by $O(q^{q-k_A k_B+c_1})$.
\end{theorem}
\begin{proof}
	\vspace{-0.15in}	
	Let $\tau = k_A k_B$ and suppose that the workers indexed by $i_0, \dots, i_{\tau -1}$ complete their tasks. Let $\bfG^A_{l}$ denote the $l$-th block column of $\bfG^A$ (with similar notation for $\bfG^B$).
	Note that for $k_1, k_2 \in \{0,1\}$ the $l$-th worker node computes $\hat{\bfA}_{\langle l,k_1 \rangle}^T \hat{\bfB}_{\langle l,k_2 \rangle}$ which can be written as
			\vspace{-0.15in}	
	\begin{align*}
&\left(\sum_{\alpha \in [k_A], \beta \in \{0,1\}} \bfG^A(2\alpha + \beta,2l + k_1) \bfA^T_{\langle \alpha,\beta \rangle}\right) \times \left(\sum_{\alpha \in [k_B], \beta \in \{0,1\}} \bfG^B(2\alpha + \beta,2l + k_2) \bfB_{\langle \alpha,\beta \rangle}\right)\\
		&\equiv \bfZ \cdot (\bfG^A(:,2l + k_1) \otimes \bfG^B(:,2l + k_2)),
	\end{align*}
%
	using the properties of the Kronecker product. Based on this, it can be observed that the decodability of $\bfZ$ at the master node is equivalent to checking whether the following matrix is full-rank.
	\begin{align*}
		\tilde{\bfG} &= [\bfG^A_{i_0} \otimes \bfG^B_{i_0} |  \bfG^A_{i_1} \otimes \bfG^B_{i_1} | \dots | \bfG^A_{i_{\tau-1}} \otimes \bfG^B_{i_{\tau-1}}].
	\end{align*}
	\noindent To analyze this matrix, consider the following decomposition of $\bfG^A_{l} \otimes \bfG^B_{l}$, for $l \in [n]$.
	\begin{align*}
		&\bfG^A_{l} \otimes \bfG^B_{l} =
		 \begin{bmatrix}
			\bfQ \bfQ^* \\
			\bfQ \Lambda^l \bfQ^*\\
			\vdots\\
			\bfQ \Lambda^{l(k_A -1)} \bfQ^*
		\end{bmatrix}
		\otimes
		\begin{bmatrix}
			\bfQ \bfQ^* \\
			\bfQ \Lambda^{l k_A} \bfQ^*\\
			\vdots\\
			\bfQ \Lambda^{l k_A(k_B -1)} \bfQ^*
		\end{bmatrix} = \\
		& \left( (\bfI_{k_A} \otimes \bfQ)
		\begin{bmatrix}
			\bfI \\
			\Lambda^{l}\\
			\vdots\\
			\Lambda^{l (k_A -1)}
		\end{bmatrix}
		\begin{bmatrix}
			\bfQ^*
		\end{bmatrix}
        \right)
		\otimes
        \left(
		(\bfI_{k_B} \otimes \bfQ)
		\begin{bmatrix}
			\bfI \\
			\Lambda^{l k_A}\\
			\vdots\\
			\Lambda^{l k_A(k_B -1)}
		\end{bmatrix}
		\begin{bmatrix}
			\bfQ^*
		\end{bmatrix} \right),
\end{align*}
where the first equality uses the eigen-decomposition of $\bfR_\theta$.
Applying the properties of Kronecker products, this can be simplified as
\begin{align*}
		\underbrace{
			\left((\bfI_{k_A} \otimes \bfQ) \otimes (\bfI_{k_B} \otimes \bfQ)\right)
		}_{\tilde{\bfQ}_1} \times
		\underbrace{\left( \begin{bmatrix}
				\bfI \\
				\Lambda^{l}\\
				\vdots\\
				\Lambda^{l (k_A -1)}
			\end{bmatrix}
			\otimes
			\begin{bmatrix}
				\bfI \\
				\Lambda^{l k_A}\\
				\vdots\\
				\Lambda^{l k_A(k_B -1)}
			\end{bmatrix}
			\right)}_{\bfX_l}
		\underbrace{
			\left(\begin{bmatrix}
				\bfQ^*
			\end{bmatrix}^{\otimes 2}
			\right)}_{\tilde{\bfQ}_2}.
	\end{align*}
Therefore, we can express
	\begin{align*}
		\tilde{\bfG}&= [\bfG^A_{i_0} \otimes \bfG^B_{i_0} |  \bfG^A_{i_1} \otimes \bfG^B_{i_1} | \dots | \bfG^A_{i_{\tau-1}} \otimes \bfG^B_{i_{\tau-1}}]\\
		&= \tilde{\bfQ}_1 [ \bfX_{i_0} | \bfX_{i_1} | \dots | \bfX_{i_{\tau-1}}]
		\begin{bmatrix}
			\tilde{\bfQ}_2 & 0 & \dots & 0\\
			0 & \tilde{\bfQ}_2 & \dots & 0\\
			\vdots & \vdots & \ddots & \vdots\\
			0 & 0 & \dots & \tilde{\bfQ}_2
		\end{bmatrix}.
	\end{align*}
	Once again, we can conclude that the invertibility and the condition number of $\tilde{\bfG}$ only depends on $[ \bfX_{i_0} | \bfX_{i_1} | \dots | \bfX_{i_{\tau-1}}]$ as the matrices pre- and post- multiplying it are both unitary. The invertibility of $[ \bfX_{i_0} | \bfX_{i_1} | \dots | \bfX_{i_{\tau-1}}]$ follows from an application of Claim \ref{claim:simplified_proof_kron_inv} in the Appendix \ref{sec:aux_claims}. The proof of Claim \ref{claim:simplified_proof_kron_inv} also shows that upon appropriate row-column permutations, the matrix $[ \bfX_{i_0} | \bfX_{i_1} | \dots | \bfX_{i_{\tau-1}}]$ can be expressed as a block-diagonal matrix with four blocks each of size $\tau \times \tau$. Each of these blocks is a Vandermonde matrix with parameters from the set $\{1, \omega_q, \omega_q^2, \dots, \omega_q^{q-1}\}$. Therefore, $[ \bfX_{i_0} | \bfX_{i_1} | \dots | \bfX_{i_{\tau-1}}]$ is non-singular and it follows that the threshold of our scheme is $k_A k_B$. An application of Theorem \ref{thm:cond_no_vand} implies that the worst-case condition number is at most $O(q^{q-\tau + c_1})$.
\end{proof}

\begin{remark}
The proofs of Theorem \ref{thm:rotmat_mv} and \ref{thm:matMatRotEmbed} involve a diagonalization argument with pre- and post-multiplying matrices that are unitary. We emphasize that this is only for the analysis of the scheme and the encoding and decoding schemes do not require multiplication by these matrices.
\end{remark}

\noindent {\bf Complexity Analysis:} Creating the $\hat{\bfA}_{\langle i, l\rangle}$ matrix requires a total of $\Delta_A$ scalar multiplications and $\Delta_A -1$ additions of block-columns of size $t \times r/\Delta_A$; a similar argument applies for creating the $\hat{\bfB}_{\langle i, l \rangle}$ matrix (note that $\Delta_A = 2 k_A, \Delta_B = 2k_B$). Thus, the total encoding complexity is given by $O((r + w)tn)$. Each worker node computes four submatrix products. Thus, the worker node computational cost is $O(4 \times rtw/\Delta_A \Delta_B) = O(rtw/k_A k_B)$. The decoding process involves inverting a matrix of dimension $\Delta_A\Delta_B \times \Delta_A \Delta_B$ followed by solving $rw/\Delta_A \Delta_B$ systems of equations. Thus, the overall decoding complexity is given by $O((\Delta_A \Delta_B)^3 + rw \Delta_A \Delta_B)$. It can be seen that the decoding complexity is independent of $t$. Thus, when the input matrices are large, i.e., $r, w$ and $t$ are large, then the overall cost is dominated by the worker node computation time.

\section{Generalized Distributed Matrix Multiplication}
\label{sec:gDMM}
In the previous section, we consider the case that $\bfA$ and $\bfB$ are partitioned into block-columns. In this section, we consider a more general scenario where $\bfA$ and $\bfB$ are partitioned into block-columns and block-rows. This construction resembles the entangled polynomial codes of \cite{yu2018straggler}.

\subsection{Matrix Splitting Scheme}
We partition the matrices $\bfA$ and $\bfB$ into $2p$ block-rows and $\Delta_A=k_A$ block-columns, and $2p$ block-rows and $\Delta_B=k_B$ block-columns respectively. We use two indices for the block-rows to simplify our presentation. In particular, we denote
\begin{align}
    \bfA &= [\bfA_{(\langle i,l\rangle, j)}], i\in [p], l \in \{0,1\}, j\in [k_A], \text{~and} \nonumber\\
    \bfB &= [\bfB_{(\langle i,l\rangle, j)}], i\in [p], l \in \{0,1\}, j\in [k_B], \label{eq:gen_matrix_subdivide}
\end{align}
where $\bfA_{(\langle i,l\rangle, j)}$ denotes the submatrix indexed by the $\langle i,l\rangle$-th block row and $j$-th block-column of $\bfA$. A similar interpretation holds for $\bfB_{(\langle i,l\rangle, j)}$. 
We let $\theta = 2\pi/q$, where $q \geq n > 2k_Ak_Bp-1$ (recall that $n$ is the number of worker nodes) is an odd integer.

The encoding in this scenario is more complicated to express. We simplify this by leveraging the following simple result which can be easily verified. 
\begin{lemma}
	\label{lemma:encodingKron}
	Suppose that matrices $\bfM_1$ and $\bfM_2$ both have $\zeta$ rows and the same column dimension.
	Consider a $2\times 2$ matrix $\mathbf{\Psi} =[\Psi_{i,j}]$, $i=0,1, j=0,1$. Then
	\begin{align*}
	\begin{bmatrix}
	\Psi_{0,0}\bfM_1+\Psi_{0,1}\bfM_2\\
	\Psi_{1,0}\bfM_1+\Psi_{1,1}\bfM_2
	\end{bmatrix}
	=(\mathbf{\Psi} \otimes \bfI_{\zeta})
	\begin{bmatrix}
	\bfM_1\\ \bfM_2
	\end{bmatrix}.
	\end{align*}
\end{lemma}

The complete encoding algorithm appears in Algorithm \ref{Alg:Gen_MM_Embedding_Scheme}.

\begin{algorithm}[t]
	\caption{Encoding scheme for generalized distributed matrix-matrix multiplication}
	\label{Alg:Gen_MM_Embedding_Scheme}
   \textbf{Input:} Matrices $\bfA$ and $\bfB$. Storage fractions $\gamma_A = 1/pk_A,\gamma_B = 1/pk_B$. Integer $\zeta = \frac{t}{2p}$.\\
   \textbf{Output:} Worker task assignment.
   \begin{algorithmic}
       \STATE Partition $\bfA$ and $\bfB$ into $2p \times \Delta_A$ and $2p \times \Delta_B$ blocks as in (\ref{eq:gen_matrix_subdivide}).
        \FOR{$k=0$ {\bfseries to} $n-1$}
        \STATE Worker $k$ is assigned
	    \begin{align*}
	    \begin{bmatrix}
	    \hat{\bfA}_{\langle k, 0 \rangle}\\
	    \hat{\bfA}_{\langle k, 1 \rangle}
	    \end{bmatrix}
	    &=\sum_{i=0}^{p-1}\sum_{j=0}^{k_A-1}(\bfR_{-\theta}^{k((j-1)p+i+1)}\otimes \bfI_{\zeta})\begin{bmatrix}
	    \bfA_{(\langle i, 0 \rangle, j)} \\ \bfA_{(\langle i, 1 \rangle, j)}
	    \end{bmatrix}, \text{~and}\\
	    \begin{bmatrix}
	    \hat{\bfB}_{\langle k, 0 \rangle}\\
	    \hat{\bfB}_{\langle k, 1 \rangle}
	    \end{bmatrix}
	    &=\sum_{i=0}^{p-1}\sum_{j=0}^{k_B-1}(\bfR_{\theta}^{k(p-1-i+jpk_A)} \otimes \bfI_{\zeta})\begin{bmatrix}
	    \bfB_{(\langle i, 0 \rangle, j)} \\ \bfB_{(\langle i, 1 \rangle, j)}
	    \end{bmatrix}.
	    \end{align*}
        \ENDFOR
        \STATE Worker $k$ computes
        \begin{align*}
        \begin{bmatrix}
        \hat{\bfA}_{\langle k, 0 \rangle}\\
        \hat{\bfA}_{\langle k, 1 \rangle}
        \end{bmatrix}^T
        \begin{bmatrix}
        \hat{\bfB}_{\langle k, 0 \rangle}\\
        \hat{\bfB}_{\langle k, 1 \rangle}
        \end{bmatrix}.
        \end{align*}
   \end{algorithmic}
\end{algorithm}

The $k$-th worker node stores $\hat{\bfA}_{\langle k,l \rangle}$, $\hat{\bfB}_{\langle k, l \rangle}$, $l = 0, 1$. Thus, each worker node stores  $\gamma_A = \frac{2}{2pk_A}=\frac{1}{pk_A}$ and $\gamma_B = \frac{2}{2pk_B}=\frac{1}{pk_B}$ fraction of $\bfA$ and $\bfB$ respectively. Worker node $k$ computes
\begin{align}
\begin{bmatrix}
\hat{\bfA}_{\langle k, 0 \rangle}\\
\hat{\bfA}_{\langle k, 1 \rangle}
\end{bmatrix}^T
\begin{bmatrix}
\hat{\bfB}_{\langle k, 0 \rangle}\\
\hat{\bfB}_{\langle k, 1 \rangle}
\end{bmatrix}. \label{eq:gen_matmat_worker_comp}
\end{align}

Before presenting our decoding algorithm and the main result of this section, we discuss the following example that helps clarify the underlying ideas.
\begin{example}
\label{eg:gen_matmat}
	Suppose $k_A = 1, k_B = 1, p = 2$. Let $n=4$. The matrix $\bfA$ and $\bfB$ can be partitioned as follows.
	\begin{align*}
	\bfA = \begin{bmatrix}
	\bfA_{(\langle 0,0\rangle, 0)}\\
	\bfA_{(\langle 0,1\rangle, 0)}\\
	\bfA_{(\langle 1,0\rangle, 0)}\\
	\bfA_{(\langle 1,1\rangle, 0)}\\
	\end{bmatrix},\text{~and~}
	\bfB = \begin{bmatrix}
	\bfB_{(\langle 0,0\rangle, 0)}\\
	\bfB_{(\langle 0,1\rangle, 0)}\\
	\bfB_{(\langle 1,0\rangle, 0)}\\
	\bfB_{(\langle 1,1\rangle, 0)}\\
	\end{bmatrix}.
	\end{align*}
	
	In this example, since $k_A=k_B=1$, there is only one block column in $\bfA$ and $\bfB$. Therefore, the index $j$ in $\bfA_{(\langle i, l\rangle, j)}$ and $\bfB_{(\langle i, l\rangle, j)}$ is always $0$. Accordingly, to simplify our presentation, we only use indices $i$ and $l$ to refer to the respective constituent block rows of $\bfA$ and $\bfB$. That is, we simplify $\bfA_{(\langle i, l\rangle, j)}$ and $\bfB_{(\langle i, l\rangle, j)}$ to  $\bfA_{\langle i, l\rangle}$ and $\bfB_{\langle i, l\rangle}$, respectively. Our scheme aims to allow the master node to recover $\bfA^T\bfB=\bfA_{\langle 0,0\rangle}^T\bfB_{\langle 0,0\rangle}+\bfA_{\langle 0,1\rangle}^T\bfB_{(\langle 0,1\rangle}+\bfA_{\langle 1,0\rangle}^T\bfB_{\langle 1,0\rangle}+\bfA_{\langle 1,1\rangle}^T\bfB_{\langle 1,1\rangle}$. Suppose that $\bfA_{\langle i, l \rangle}$ and $\bfB_{\langle i, l \rangle}$ have $\zeta$ rows. The encoding process ({\it cf.} Algorithm \ref{Alg:Gen_MM_Embedding_Scheme}) can be defined as
	\begin{align*}
		\begin{bmatrix}
			\hat{\bfA}_{\langle k, 0 \rangle}\\
			\hat{\bfA}_{\langle k, 1 \rangle}
		\end{bmatrix}
		&=\sum_{i=0}^1 (\bfR_{-\theta}^{k(i-1)}\otimes \bfI_{\zeta})\begin{bmatrix}
		\bfA_{\langle i, 0 \rangle} \\ \bfA_{\langle i, 1 \rangle}
		\end{bmatrix}, \text{~and}\\
		\begin{bmatrix}
		\hat{\bfB}_{\langle k, 0 \rangle}\\
		\hat{\bfB}_{\langle k, 1 \rangle}
		\end{bmatrix}
		&=\sum_{i=0}^1 (\bfR_\theta^{k(1-i)}\otimes \bfI_{\zeta})\begin{bmatrix}
		\bfB_{\langle i, 0 \rangle} \\ \bfB_{\langle i, 1 \rangle}
		\end{bmatrix}.
	\end{align*}
The computation in worker node $k$ ({\it cf.} (\ref{eq:gen_matmat_worker_comp})) can be analyzed as follows. Let
$\begin{bmatrix}
	\bfA^\calF_{\langle i, 0 \rangle}\\
	\bfA^\calF_{\langle i, 1 \rangle}
	\end{bmatrix}=(\bfQ^*\otimes \bfI_{\zeta})\begin{bmatrix}
	\bfA_{\langle i, 0 \rangle}\\
	\bfA_{\langle i, 1 \rangle}
	\end{bmatrix}$
	and
	$\begin{bmatrix}
	\bfB^\calF_{\langle i, 0 \rangle}\\
	\bfB^\calF_{\langle i, 1 \rangle}
	\end{bmatrix}=(\bfQ^*\otimes \bfI_{\zeta})\begin{bmatrix}
	\bfB_{\langle i, 0 \rangle}\\
	\bfB_{\langle i, 1 \rangle}
	\end{bmatrix}$. Then
	
	\begin{align*}
	\begin{bmatrix}
	\hat{\bfA}_{\langle k, 0 \rangle}\\
	\hat{\bfA}_{\langle k, 1 \rangle}
	\end{bmatrix}^T
	\begin{bmatrix}
	\hat{\bfB}_{\langle k, 0 \rangle}\\
	\hat{\bfB}_{\langle k, 1 \rangle}
	\end{bmatrix} \overset{(a)}{=}&  \bigg((\bfQ^*\otimes \bfI_{\zeta})\begin{bmatrix}
	\hat{\bfA}_{\langle k, 0 \rangle}\\
	\hat{\bfA}_{\langle k, 1 \rangle}
	\end{bmatrix}\bigg)^* (\bfQ^* \otimes \bfI_{\zeta})\begin{bmatrix}
	\hat{\bfB}_{\langle k, 0 \rangle}\\
	\hat{\bfB}_{\langle k, 1 \rangle}
	\end{bmatrix}\\
	=&\bigg((\bfQ^*\otimes \bfI_{\zeta})(\bfR_{-\theta}^{-k}\otimes \bfI_{\zeta})\begin{bmatrix}
	\bfA_{\langle 0, 0\rangle}\\
	\bfA_{\langle 0, 1\rangle}
	\end{bmatrix}+(\bfQ^*\otimes \bfI_{\zeta})(\bfI_2\otimes \bfI_{\zeta})\begin{bmatrix}
	\bfA_{\langle 1, 0\rangle}\\
	\bfA_{\langle 1, 1\rangle}
	\end{bmatrix}\bigg)^*\\
	&\bigg((\bfQ^*\otimes \bfI_{\zeta})(\bfR_{\theta}^{k}\otimes \bfI_{\zeta})\begin{bmatrix}
	\bfB_{\langle 0, 0\rangle}\\
    \bfB_{\langle 0, 1\rangle}
	\end{bmatrix}+(\bfQ^* \otimes \bfI_{\zeta})(\bfI_2\otimes \bfI_{\zeta})\begin{bmatrix}
	\bfB_{\langle 1, 0\rangle}\\
	\bfB_{\langle 1, 1\rangle}
	\end{bmatrix}\bigg)\\
	 \overset{(b)}{=}&\bigg((\bfQ^*\bfR_{-\theta}^{-k}\bfQ\otimes \bfI_{\zeta})(\bfQ^*\otimes \bfI_{\zeta})\begin{bmatrix}
	\bfA_{\langle 0, 0\rangle}\\
	\bfA_{\langle 0, 1\rangle}
	\end{bmatrix}+(\bfQ^*\bfI_2\bfQ\otimes \bfI_{\zeta})(\bfQ^*\otimes \bfI_{\zeta})\begin{bmatrix}
    \bfA_{\langle 1, 0\rangle}\\
	\bfA_{\langle 1, 1\rangle}
	\end{bmatrix}\bigg)^*\\
	&\bigg((\bfQ^*\bfR_{\theta}^{k}\bfQ\otimes \bfI_{\zeta})(\bfQ^*\otimes \bfI_{\zeta})\begin{bmatrix}
	\bfB_{\langle 0, 0\rangle}\\
	\bfB_{\langle 0, 1\rangle}
	\end{bmatrix} + (\bfQ^*\bfI_2\bfQ \otimes \bfI_{\zeta})(\bfQ^* \otimes \bfI_{\zeta})\begin{bmatrix}
	\bfB_{\langle 1, 0\rangle}\\
	\bfB_{\langle 1, 1\rangle}
	\end{bmatrix}\bigg)\\
	\overset{(c)}{=}&\bigg(\left(\begin{bmatrix}
	{\omega_q^*}^{-k} & 0\\
	0 & {\omega_q^*}^{k}
	\end{bmatrix}\otimes \bfI_{\zeta}\right)(\bfQ^*\otimes \bfI_{\zeta})\begin{bmatrix}
	\bfA_{\langle 0, 0\rangle}\\
	\bfA_{\langle 0, 1\rangle}
	\end{bmatrix}+\left(\begin{bmatrix}
	1 & 0\\
	0 & 1
	\end{bmatrix}\otimes \bfI_{\zeta}\right)(\bfQ^* \otimes \bfI_{\zeta})\begin{bmatrix}
    \bfA_{\langle 1, 0\rangle}\\
	\bfA_{\langle 1, 1\rangle}
	\end{bmatrix}\bigg)^*\\
	&\bigg( \left(\begin{bmatrix}
	{\omega_q}^{k} & 0\\
	0 & {\omega_q}^{-k}
	\end{bmatrix}\otimes \bfI_{\zeta}\right)(\bfQ^* \otimes \bfI_{\zeta})\begin{bmatrix}
	\bfB_{\langle 0, 0\rangle}\\
	\bfB_{\langle 0, 1\rangle}
	\end{bmatrix}+\left(\begin{bmatrix}
	1 & 0\\
	0 & 1
	\end{bmatrix}\otimes \bfI_{\zeta}\right)(\bfQ^* \otimes \bfI_{\zeta})\begin{bmatrix}
	\bfB_{\langle 1, 0\rangle}\\
	\bfB_{\langle 1, 1\rangle}
	\end{bmatrix}\bigg)\\
	\overset{(d)}{=}&\bigg(\begin{bmatrix}
	{\omega_q^*}^{-k}\bfA_{\langle 0, 0\rangle}^\calF\\
	{\omega_q^*}^{k}\bfA_{\langle 0, 1\rangle}^\calF
	\end{bmatrix}+\begin{bmatrix}
	\bfA_{\langle 1, 0\rangle}^\calF\\
	\bfA_{\langle 1, 1\rangle}^\calF
	\end{bmatrix}\bigg)^*\bigg(\begin{bmatrix}
	{\omega_q}^{k}\bfB_{\langle 0, 0\rangle}^\calF\\
	{\omega_q}^{-k}\bfB_{\langle 0, 1\rangle}^\calF
	\end{bmatrix}+\begin{bmatrix}
	\bfB_{\langle 1, 0\rangle}^\calF\\
	\bfB_{\langle 1, 1\rangle}^\calF
	\end{bmatrix}\bigg)\\
	=&(\bfA^{\calF *}_{\langle 0,0\rangle}\bfB^{\calF}_{\langle 1,0\rangle}+\bfA^{\calF *}_{\langle 1,1\rangle}\bfB^{\calF}_{\langle 0,1\rangle})\omega_q^{-k}+\\
	&(\bfA^{\calF *}_{\langle 0,0\rangle}\bfB^{\calF}_{\langle 0,0\rangle}+\bfA^{\calF *}_{\langle 1,0\rangle}\bfB^{\calF}_{\langle 1,0\rangle}+\bfA^{\calF *}_{\langle 0,1\rangle}\bfB^{\calF}_{\langle 0,1\rangle}+ \bfA^{\calF *}_{\langle 1,1\rangle}\bfB^{\calF}_{\langle 1,1\rangle})+\\
	&(\bfA^{\calF *}_{\langle 1,0\rangle}\bfB^{\calF}_{\langle 0,0\rangle}+\bfA^{\calF *}_{\langle 0,1\rangle}\bfB^{\calF}_{\langle 1,1\rangle})\omega_q^{k}
	\end{align*}
    where
    \begin{itemize}
    	\item $(a)$ holds because $\bfQ^* \otimes \bfI_\zeta$ is unitary,
    	\item $(b)$ holds by the mixed-product property of Kronecker product. For example,
    	\begin{align*}
    	(\bfQ^* \otimes \bfI_\zeta)(\bfR_{-\theta}^{-k}\otimes \bfI_\zeta) &= (\bfQ^*\bfR_{-\theta}^{-k})\otimes \bfI_\zeta\\
    	&=(\bfQ^*\bfR_{-\theta}^{-k}\bfQ \bfQ^*)\otimes \bfI_\zeta\\
    	&=(\bfQ^*\bfR_{-\theta}^{-k}\bfQ \otimes \bfI_\zeta)(\bfQ^*\otimes\bfI_{\zeta}).
    	\end{align*}
    	\item 
    	$(c)$ holds because $\bfQ^*\bfR_{\theta}\bfQ=\begin{bmatrix}
    	\omega_q & 0\\
    	0 &\omega_q^{-1}
    	\end{bmatrix}$, and
    	\item $(d)$ holds by Lemma \ref{lemma:encodingKron}.
    \end{itemize}

    Thus, it is clear that whenever the master node collects the results of any three distinct worker nodes, it can recover $\bfA^{\calF *}_{\langle 0,0\rangle}\bfB^{\calF}_{\langle 0,0\rangle}+\bfA^{\calF *}_{\langle 1,0\rangle}\bfB^{\calF}_{\langle 1,0\rangle}+\bfA^{\calF *}_{\langle 0,1\rangle}\bfB^{\calF}_{\langle 0,1\rangle}+ \bfA^{\calF *}_{\langle 1,1\rangle}\bfB^{\calF}_{\langle 1,1\rangle}$.
    However, we observe that for $i = 0,1$
    \begin{align*}
    \begin{bmatrix}
    \bfA^\calF_{\langle i, 0 \rangle}\\
    \bfA^\calF_{\langle i, 1 \rangle}
    \end{bmatrix}^*
    \begin{bmatrix}
    	\bfB^\calF_{\langle i, 0 \rangle}\\
    	\bfB^\calF_{\langle i, 1 \rangle}
    \end{bmatrix}=\begin{bmatrix}
    \bfA_{\langle i, 0 \rangle}\\
    \bfA_{\langle i, 1 \rangle}
    \end{bmatrix}^T\begin{bmatrix}
    \bfB_{\langle i, 0 \rangle}\\
    \bfB_{\langle i, 1 \rangle}
    \end{bmatrix}.
    \end{align*}
    Thus, we can equivalently recover $\bfA^T \bfB$.
\end{example}
The analysis in the example above can be generalized to show the following result. The proof appears in the Appendix \ref{sec:proof_genMatMat}.
\begin{theorem}
	\label{thm:genMatMat}
	The threshold for scheme in this section is $2pk_Ak_B - 1$. The worst-case condition number of the recovery matrices is upper bounded by $O(q^{q-2pk_Ak_B+1 + c_1})$.
\end{theorem}

\begin{remark}
When $k_A = k_B = 1$, the threshold of this scheme matches the Entangled Polynomial code \cite{yu2018straggler} and the MatDot codes \cite{dutta2019optimal}, with the added advantage of excellent numerical stability.
\end{remark}
The decoding algorithm in this case requires more steps. It is specified in Algorithm \ref{Alg:Gen_MM_decoding_Embedding_Scheme}. In particular, it requires us to work with the inverse of a complex matrix (see (\ref{eq:g_vand_I})) which is essentially (upto a unitary scaling) a Vandermonde matrix with parameters on the unit circle. The underlying reason can be found by examining the proof of Theorem \ref{thm:genMatMat}. Thus, the decoding in this case is more expensive than prior methods that work exclusively with real valued decoding. Nevertheless, we emphasize that the worker node computation is still real-valued.

Suppose that the $k$-th worker node computes $\hat{\bfA}_k^T\hat{\bfB}_k$ and that the master node receives the computation results from any $\tau = 2pk_Ak_B-1$ worker nodes, which are denoted by ${i_0}, \cdots, {i_{\tau-1}}$. By (\ref{eq:generalize_sum}), the useful and interference terms can be decoded by computing the inverse of
\begin{align}
\bfG_{\calI}^{vand}=
\begin{bmatrix}
\omega_q^{-i_0(pk_Ak_B-1)} & \omega_q^{-i_1(pk_Ak_B-1)} &
\cdots &
\omega_q^{-i_{\tau-1}(pk_Ak_B-1)}\\
\omega_q^{-i_0(pk_Ak_B-2)} & \omega_q^{-i_1(pk_Ak_B-2)} &
\cdots &
\omega_q^{-i_{\tau-1}(pk_Ak_B-2)}\\
\vdots & \vdots & \ddots & \vdots\\
1 & 1 & \cdots & 1\\
\vdots & \vdots & \ddots & \vdots\\
\omega_q^{i_0(pk_Ak_B-1)} & \omega_q^{i_1(pk_Ak_B-1)} &
\cdots &
\omega_q^{i_{\tau-1}(pk_Ak_B-1)}
\end{bmatrix}. \label{eq:g_vand_I}
\end{align}
and using it to solve $\frac{r}{k_A} \times \frac{w}{k_B}$ systems of equations.
We point out that by multiplying $\bfG_{\calI}^{vand}$ from the right by the unitary matrix $[\text{diag}(\omega_q^{i_0}, \dots, \omega_q^{i_{\tau-1}})]^{pk_Ak_B-1}$, it can be seen that $\kappa(\bfG_{\calI}^{vand})$ is the same as the condition number of a Vandermonde matrix of size $(2pk_Ak_B-1) \times (2pk_Ak_B-1)$ with parameters $\omega_q^{i_0}, \dots, \omega_q^{i_{\tau-1}}$.

Finally, the result $\bfC = [\bfC_{i,j}]$, $i\in [k_A]$, $j\in [k_B]$ can be recovered since $\bfC_{i,j} = \sum_{u=0}^{p-1} (\bfA_{(\langle u, 0 \rangle, i)}^T \bfB_{(\langle u, 0 \rangle, j)} + \bfA_{(\langle u, 1 \rangle, i)}^T \bfB_{(\langle u, 1 \rangle, j)})=\left( \sum_{u=0}^{p-1} (\bfA_{(\langle u, 0 \rangle, i)}^{\calF *} \bfB_{(\langle u, 0 \rangle, j)}^{\calF}\right) + \left(\sum_{u=0}^{p-1} \bfA_{(\langle u, 1 \rangle, i)}^{\calF *} \bfB_{(\langle u, 1 \rangle, j)}^{\calF})\right)$.
The precise decoding algorithm is summarized in Algorithm \ref{Alg:Gen_MM_decoding_Embedding_Scheme}.

\begin{algorithm}[t]
	\caption{Decoding scheme for generalized distributed matrix-matrix  multiplication}
	\label{Alg:Gen_MM_decoding_Embedding_Scheme}
	\textbf{Input:} $\bfG_{\calI}^{vand}$  ({\it cf.} (\ref{eq:g_vand_I})) where $|\calI|=2pk_Ak_B-1$ (columns of $\bfG^{vand}$ corresponding to columns in $\calI$). Row vectors $\bfc$ corresponding to the observed values in each of the $\frac{r}{k_A} \times \frac{w}{k_B}$  system of equations.\\ 
	%
	%
\textbf{Output:} Decoded estimate $\tilde{\bfC}$ of $\bfA^T \bfB$.
\begin{algorithmic}
\STATE  {\bfseries 1. procedure:} Decode $\hat{\bfm}$ from $\bfc$
\STATE $\hat{\bfm} = [\hat{\bfm}_{-pk_Ak_B}, \cdots, \hat{\bfm}_0, \cdots, \hat{\bfm}_{pk_Ak_B} ]$ by $\hat{\bfm}=\bfc (\bfG_{\calI}^{vand})^{-1}$.
		\STATE {\bfseries end procedure}
\STATE {\bfseries 2. procedure:} Repeat above procedure for each of the $\frac{r}{k_A} \times \frac{w}{k_B}$ systems of equations. Upon appropriate indexing, we can form a matrix $\hat{\bfM}_{i,j}, -(k_A -1) \leq i\leq k_A-1, -(k_B -1) \leq j \leq k_B-1$ using the decoded components $\hat{\bfm}_{ip + jpk_A}$.
        \STATE {\bfseries end procedure}
\STATE  {\bfseries 3. procedure:} Recover $\tilde{\bfC}_{i_1,j_1}$ for $i_1 \in [k_A], j_1 \in [k_B]$.
 \IF{$i_1=0,j_1=0$}
 \STATE $\tilde{\bfC}_{0,0} = \hat{\bfM}_{0,0}$. 
 \ELSE
\STATE  $\tilde{\bfC}_{i_1,j_1} = \hat{\bfM}_{i_1,j_1}+\hat{\bfM}_{-i_1,-j_1}$.
\ENDIF

		\STATE {\bfseries end procedure}
\end{algorithmic}
\end{algorithm}
\noindent {\bf Complexity Analysis:} We note here that the decoding algorithm involving inverting a $(2pk_Ak_B-1)\times (2pk_Ak_B-1)$ complex Vandermonde matrix once and using the inverse to solve $\frac{r}{k_A}\times \frac{w}{k_B}$ systems of equations in Steps 1 and 2. Step 3 involves the sum of matrices of size $\frac{r}{k_A} \times \frac{w}{k_B}$ so its complexity is $O(rw)$. Thus, the overall decoding complexity is $O((2pk_Ak_B-1)^3 + rw + \frac{rw}{k_Ak_B}(2pk_Ak_B-1)^2)\approx O(p^3k_A^3k_B^3+ rwp^2k_Ak_B)$, where typically, $rw \gg pk_A^2k_B^2$.


\section{Comparisons and Numerical Experiments}
\label{sec:comps}
We now present a comparison of our techniques with other approaches in the literature. Towards this end we will compare the worst-case and the average condition numbers of the recovery matrices of the different schemes. Furthermore, we will also present corresponding normalized mean-squared-error (MSE) vs. SNR curves. For matrix-vector multiplication, let $\bfA^T \bfx$ denote the true value of the computation and $\widehat{\bfA^T \bfx}$ denote the result of using one of the discussed methods. The normalized MSE is defined as $\frac{||\bfA^T \bfx - \widehat{\bfA^T \bfx}||_{F}}{||\bfA^T \bfx||_F}$ (the notation $||\cdot||_F$ denotes the Frobenius norm of the matrix). Similarly, for the matrix-matrix multiplication, the normalized MSE is given by $\frac{||\bfA^T \bfB - \widehat{\bfA^T \bfB}||_{F}}{||\bfA^T \bfB||_F}$ where $\bfA^T \bfB$ is the true product and $\widehat{\bfA^T \bfB}$ is the decoded product using one of the methods. We will also report the computation threshold, worker computation times and decoding times for all the methods under consideration.

Suppose that the number of workers $n$ is odd, so that we can pick $q=n$ for the rotation matrix embedding. From a theoretical perspective our schemes have a worst-case condition number (over the different recovery submatrices) that is upper bounded by $O(q^{q -\tau +c_1})$ where $\tau$ is the recovery threshold. Equivalently, the worst-case condition number is upper bounded by $O(n^{s+c_1})$ (recall that $c_1 = 5.5$). We note here that this upper bound is definitely loose and our numerical experiments which will be presented shortly indicate that  the actual condition number values are much smaller. The work of \cite{FahimC19} shows a condition number upper bound which is $O(n^{2s})$. While this is larger than our upper bound for values of $s \geq 6$ we emphasize that our actual condition number values are much lower than \cite{FahimC19} even for $s \leq 6$.

As discussed previously, the scheme of \cite{yu2017polynomial} has condition numbers that are exponential in the recovery threshold $\tau$. This is corroborated by our numerical experiments as well. In Section VII of \cite{yu2018straggler}, the authors propose a finite field embedding approach  as a potential solution to the numerical issues encountered when operating over the reals. For this purpose the real entries will need to multiplied by large enough integers and then quantized so that each entry lies with $0$ and $p-1$ for a large enough prime $p$. All computations will be performed within the finite field of order $p$, i.e., by reducing the computations modulo-$p$. This technique requires that each $\bfA_i^T \bfB_j$ needs to have all its entries within $0$ to $p-1$, otherwise there will be errors in the computation. Let $\alpha$ be an upper bound on the absolute value of matrix entries in $\bfA$ and $\bfB$. Then, this means that the following dynamic range constraint (DRC),
\begin{align*}
\alpha^2 t \leq p-1
\end{align*}
needs to be satisfied. Otherwise, the modulo-$p$ operation will cause arbitrarily large errors.

We note here that the publicly available code for \cite{yu2017polynomial} uses $p=65537$. Now consider a system with $k_A = 3$, $k_B = 2$. Even for small matrices with $\bfA$ of size $400\times 200$, $\bfB$ of size $400\times 300$ and entries chosen as random integers between $0$ to $30$, the DRC is violated for $p=65537$ since $30^2 \times 400 > 65537$. In this scenario, the normalized MSE of the \cite{yu2017polynomial} approach is $0.7746$. In contrast, our method has a normalized MSE $\approx 2\times 10^{-28}$ for the same system with $k_A=3, k_B=2$.

When working over $64$-bit integers, the largest integer is $\approx 10^{19}$. Thus, even if $t \approx 10^5$, the finite-field embedding method can only support $\alpha \leq 10^7$. Thus, the range is rather limited. Furthermore, considering matrices of limited dynamic range is not a valid assumption. In machine learning scenarios such as deep neural networks, matrix multiplications are applied repeatedly, and the output of one stage serves as the input for the other. Thus, over several iterations the dynamic range of the matrix entries will grow. Consequently, applying this technique will necessarily incur quantization error.

The most serious limitation of the method comes from the fact the error in the computation (owing to quantization) is strongly dependent on the actual entries of the $\bfA$ and $\bfB$ matrices. In fact, we can generate structured integer matrices $\bfA$ and $\bfB$ such that the normalized MSE of their approach is exactly $1.0$. Towards this end we first pick the prime  $p=2147483647$ (which is much larger than their publicly available code) so that their method can support higher dynamic range. 
Next let $r=w=t=2000$. 
This implies that $\alpha$ has to be $\leq 1000$ by the dynamic range constraint. For $k_A = k_B =2$, the matrices have the following block decomposition.
\begin{align*}
\bfA
&= \begin{bmatrix}
\bfA_{0,0} & \bfA_{0,1}\\
\bfA_{1,0} & \bfA_{1,1}
\end{bmatrix}, \text{~and~}\\
\bfB &= \begin{bmatrix}
\bfB_{0,0} & \bfB_{0,1}\\
\bfB_{1,0} & \bfB_{1,1}
\end{bmatrix}.
\end{align*}
Each $\bfA_{i,j}$ and $\bfB_{i,j}$ is a matrix of size $1000 \times 1000$, with entries chosen from the following distributions. $\bfA_{0,0}$, $\bfA_{0,1}$ are distributed $\text{Unif}(0, \dots,9999)$ and $\bfA_{1,0}$, $\bfA_{1,1}$ distributed $\text{Unif}(0, \dots,9)$. Next, $\bfB_{0,0}$, $\bfB_{0,1}$ are distributed $\text{Unif}(0, \dots,9)$ and $\bfB_{1,0}, \bfB_{1,1}$ distributed $\text{Unif}(0, \dots,9999)$. In this scenario, the DRC requires us to multiply each matrix by $0.1$ and quantize each entry between $0$ and $999$. Note that this implies that $\bfA_{1,0}, \bfA_{1,1}, \bfB_{0,0}, \bfB_{0,1}$ are all quantized into zero submatrices since the entry in these four submatrices is less than $10$. We label the quantized matrices by the superscript $\tilde{\cdot}$. We emphasize that the finite field embedding technique {\it only} recovers the product of these quantized matrices. However, this product is
\begin{align*}
\tilde{\bfA}^T\tilde{\bfB} = \begin{bmatrix}
\tilde{\bfA}_{0,0} & \tilde{\bfA}_{0,1}\\
\mathbf{0} & \mathbf{0}
\end{bmatrix}^T \begin{bmatrix}
\mathbf{0} & \mathbf{0}\\
\tilde{\bfB}_{1,0} & \tilde{\bfA}_{1,1}
\end{bmatrix}=\mathbf{0}.
\end{align*}
Thus, the final estimate of the original product $\bfA^T\bfB$, denoted as $\widehat{\bfA^T\bfB}$ is the all-zeros matrix. This implies that the normalized MSE of their scheme is exactly $1.0$. Thus, the performance of the finite field embedding technique has a strong dependence on the matrix entries. We note here that even if we consider other quantization schemes or larger 64-bit primes, one can arrive at adversarial examples such as the ones shown above. Once again for these examples, our methods have a normalized MSE of at most $10^{-27}$.

\begin{table*}[t]
	\centering
	\caption{Performance of matrix inversion over a large prime order field in Python 3.7. The table shows the computation time for inverting a $\ell\times \ell$ matrix $\bfG$ over a finite field of order $p$. Let $ \widehat{\bfG^{-1}}$ denote the inverse obtained by applying the sympy function {\tt Matrix}($\bfG$) {\tt .inverse\_mod(p)}. 
The MSE is defined as  $\frac{1}{\ell}||\bfG \widehat{\bfG^{-1}} - \bfI||_{F}$. }
	\label{table:comp_mse_inverse}
	\begin{tabular}{|c|c|c|c|}
		\hline
		$\ell$ &$p$&{\small Computation Time (s)} & {\small MSE} \\
		\hline
		$9$ & $65537$ & $1.39$ & $0$\\
		\hline
		$12$ & $65537$ & $4.38$ & $0$\\
		\hline
		$15$ & $65537$ & $12.64$ & $0$\\
		\hline
		$9$ & $2147483647$ & $1.39$ & $0$\\
		\hline
		$12$ & $2147483647$ & $4.68$ & $1.8\times 10^9$\\
		\hline
		$15$ & $2147483647$ & $14.45$ & $4.2\times 10^9$\\
		\hline
	\end{tabular}
\end{table*}

In our experience, the finite field embedding technique also suffers from significant computational issues in implementation. Note that the technique requires the computation of the inverse matrix at the master node that is required for decoding the final result. We implemented this within the Python 3.7, sympy library (see \cite{inverseGit} Git hub repository). We performed experiments with $p=65537$ and $p=2147483647$. As shown in Table \ref{table:comp_mse_inverse}, for the smaller prime $p=65537$, the inverse computation is accurate up to $15 \times 15$ matrices; however, the computation time of the inverse is rather high and can dominate the overall execution time. On the other hand for the larger prime $p=2147483647$, the  error in in the computed inverse is very high for $12 \times 12$ and $15 \times 15$ matrices; the corresponding time taken is even higher. It is possible that very careful implementations can perhaps avoid these issues. However, we are unaware of any such publicly available code. To summarize, the finite field embedding technique suffers from major dynamic range limitations and associated computational issues and cannot be used to support real computations.

The work most closely related to ours is by \cite{FahimC19}, which demonstrates an upper bound of $O(q^{2(q -\tau)})$ on the worst-case condition number. It can be noted that this grows much faster than our upper bound in the parameter $q-\tau$. In numerical experiments, our worst-case condition numbers are much smaller than the work of \cite{FahimC19}; we discuss this in the upcoming Section \ref{sec:num_exp}. We note that the results in \cite{FahimC19} are given in terms of the condition number calculated using the Frobenius norm \footnote{For measuring the error in decoding a system of equations corresponding to $\bfM$ it is more natural to consider an induced norm, like the one we use.}, i.e., for matrix $\bfM$, they define $\kappa(\bfM) = ||\bfM||_{F} ||\bfM^{-1}||_F$. However, there are well-known relations between different matrix norms. In particular when $\bfM$ is of size $\ell \times \ell$, then $|| \bfM ||_2 \leq || \bfM ||_F \leq  \sqrt{\ell} || \bfM ||_2$. This allows us to compare the corresponding Frobenius-norm induced condition number as well.

Both our scheme and \cite{FahimC19} have the optimal threshold when $\bfA$ and $\bfB$ are only divided into block-columns ({\it cf.} Section \ref{sec:struc_matrices})). However, when the matrices are split across both rows and columns ({\it cf.} Section \ref{sec:gDMM}) the polynomial code approach of \cite{yu2018straggler} has a lower threshold of $pk_A k_B + p-1$, while our threshold is $2pk_Ak_B -1$; the thresholds match when $k_A = k_B = 1$.
The work of \cite{FahimC19} in this scenario, i.e., when $p > 1$ has a threshold denoted $\tau_{F-C}$ given by
\begin{align*}
\tau_{F-C} = 4k_Ak_Bp-2(k_Ak_B+pk_A+pk_B)+k_A+k_B+2p-1.
\end{align*}
It can be seen that if $k_A = 1 $ or $k_B = 1$, then $\tau_{F-C} \leq 2pk_Ak_B - 1$. However, when $k_A > 1$ and $k_B > 1$, simple analysis shows that our threshold $\leq \tau_{F-C}$ (see Claim \ref{claim:threshold_gen_new} in the Appendix).



Certain approaches \cite{mallick2019, DasTR18, DasR19, RamamoorthyTV19} only apply for matrix-vector multiplication and furthermore do not provide any explicit guarantees on the worst-case condition number.
Other approaches include the work of \cite{subramaniam2019random} which uses random linear encoding of the $\bfA$ and $\bfB$ matrices and the work of \cite{DasRV19} that uses a convolutional coding approach to this problem. Both these approaches require random sampling and do not have a theoretical upper bound on the worst-case condition number. However, for a given set of random choices, it is possible to numerically compute an upper bound on the worst-case condition number of \cite{DasRV19}. 

\begin{table*}[t]
	\centering
	\caption{Comparison for matrix-vector case with $n=31$, $\bfA$ has size $28000\times 19720$ and $\bfx$ has length $28000$ for the first four methods. For the All Ones Conv. and Random Conv. (from \cite{DasRV19}), $\bfA$ has 21924 columns.}
	\label{table:mat_vac_cond}
		\begin{tabular}{|c|c|c|c|c|c|c|}
			\hline
			{\small Scheme} &$\gamma_A$&$\tau$& {\small Avg. Cond. Num.} & {\small Max. Cond. Num.} & {\small Avg. Worker Comp. Time (s)} & {\small Dec. Time (s)} \\
			\hline
			Real Vand.&$1/29$& $29$ & $1.1\times 10^{13}$ & $2.9\times 10^{13}$ & $1.2\times 10^{-3}$& $9\times 10^{-5}$\\
			\hline
			Complex Vand.&$1/29$& $29$ & $12$ & $55$ & $2.9\times 10^{-3}$ & $2.8\times 10^{-4}$\\
			\hline
			Circ. Perm. Embed.&$1/28$& $29$ & $12$ & $55$ & $1.2\times 10^{-3}$ & $3.7\times 10^{-4}$\\
			\hline
			Rot. Mat. Embed.&$1/29$& $29$ & $12$ & $55$ & $1.3\times 10^{-3}$ & $10^{-4}$ \\
			\hline
			All Ones Conv. \cite{DasRV19} &$1/27$& $29$ & $1386$ & $5093$ & $1.4 \times 10^{-3}$ & $9 \times 10^{-4}$\\
			\hline
            Random Conv. \cite{DasRV19} &$1/27$& $29$ & $259$ & $4903$ & $1.4 \times 10^{-3}$ & $5 \times 10^{-4}$\\
            \hline
	\end{tabular}
\end{table*}

\begin{figure}[t]
	\centering
	\includegraphics[width=100mm]{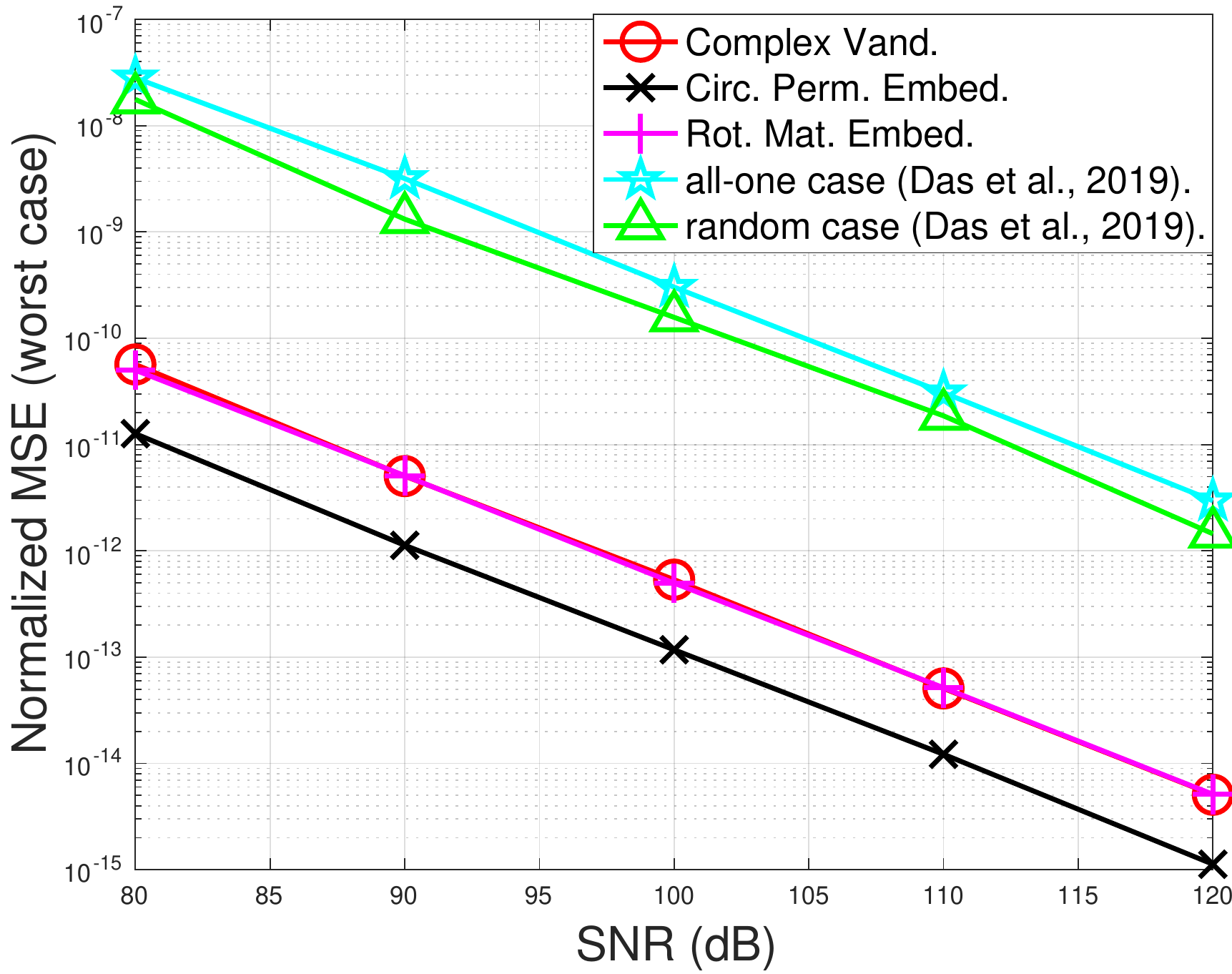}
	\caption{ {\small Consider matrix-vector $\bfA^T\bfx$ multiplication system with $n=31$, $\tau=29$. $\bfA$ has size $28000\times 19720$ and $\bfx$ has length $28000$.}}
	\label{fig:mat_vec_mul_error_rate}
\end{figure}

\subsection{Numerical Experiments}
\label{sec:num_exp}
The central point of our paper is that we can leverage the well-conditioned behavior of Vandermonde matrices with parameters on the unit circle while continuing to work with computation over the reals. We compare our results with the work of \cite{yu2017polynomial} (called ``Real Vandermonde''), a ``Complex Vandermonde'' scheme where the evaluation points are chosen from the complex unit circle, the work of \cite{DasRV19},\cite{FahimC19} and \cite{subramaniam2019random}. For the normalized MSE simulations below, we always pick the set of worker nodes that correspond to the worst-case condition number of the corresponding method. Additive Gaussian noise is added to the encoded matrix and vector in the matrix-vector case and both encoded matrices in the matrix-matrix case (details in \cite{stable2020Github}).

All experiments were run on the AWS EC2 system with a t2.2xlarge instance (for master node) and t2.micro instances (for slave nodes). The source code can be found in \cite{stable2020Github}.

\subsubsection{Matrix-vector case}
In Table \ref{table:mat_vac_cond}, we compare the {\it average} and {\it worst-case} condition number of the different schemes for matrix-vector multiplication. The system under consideration has $n=31$ worker nodes and a threshold specified by the third column (labeled as $\tau$). The evaluation points for \cite{yu2017polynomial} were uniformly sampled from the interval $[-1, 1]$ \cite{berrut2004}. The Complex Vandermonde scheme has evaluation points which are the 31-st root of unity. The \cite{FahimC19} and \cite{subramaniam2019random} schemes are not applicable for the matrix-vector case. It can be observed from Table \ref{table:mat_vac_cond} that both the worst-case and the average condition numbers of our scheme are over eleven orders of magnitude better than the Real Vandermonde scheme. Furthermore, there is an exact match of the condition number values for all the other schemes. This can be understood by following the discussion in Section \ref{sec:mat_vec_mul}. Specifically, our schemes have the property that the condition number only depends on the eigenvalues of corresponding circulation permutation matrix and rotation matrix respectively. These eigenvalues lie precisely within 31-th roots of unity. The methods of \cite{DasRV19} have some divisibility constraints on the number of columns in $\bfA$. Accordingly, we considered a matrix with 21924 columns for it. We performed 200 random trials for picking the best Random Conv. code \cite{DasRV19}. The worst-case condition number of these methods are still around one to two orders of magnitude higher than ours.

\begin{table*}[t]
	\centering
	\caption{Comparison for $\bfA^T\bfB$ matrix-matrix multiplication case with $n=31$,  $k_A=4$, $k_B=7$. $\bfA$ has size $8000\times 14000$, $\bfB$ has size $8400\times 14000$.}
	\label{table:mat_mat_cond}
	\begin{tabular}{|c|c|c|c|c|c|c|c|}
		\hline
		{\small Scheme} &$\gamma_A$&$\gamma_B$&$\tau$& {\small Avg. Cond. Num.} & {\small Max. Cond. Num.} & {\small Avg. Worker Comp. Time (s)} & {\small Dec. Time (s)} \\
		\hline
		Real Vand.&$1/4$& $1/7$&$28$ & $4.9\times 10^{12}$ & $2.3\times 10^{13}$ &$2.132$ &$0.407$ \\
		\hline
		Complex Vand.&$1/4$& $1/7$& $28$ & $27$ & $404$ & $8.421$& $1.321$\\
		\hline
		Rot. Mat. Embed.&$1/4$& $1/7$& $28$ & $27$ & $404$ & $2.121$ & $0.408$\\
		\hline
		Ortho-Poly \cite{FahimC19}&$1/4$& $1/7$& $28$ & $1449$ &  $8.3\times 10^4$ & $2.263$& $0.412$ \\
		\hline
		RKRP \cite{subramaniam2019random}&$1/4$& $1/7$ & $28$ & $255$ & $5.6\times 10^4$ & $2.198$ & $0.406$\\
        \hline
		Random Conv. \cite{DasRV19} &$1/3$& $1/6$ & $28$ & - & $\leq 3.4\times 10^4$ & - & -\\

		\hline
	\end{tabular}
\end{table*}

\begin{figure}[t]
	\centering
	\includegraphics[width=70mm]{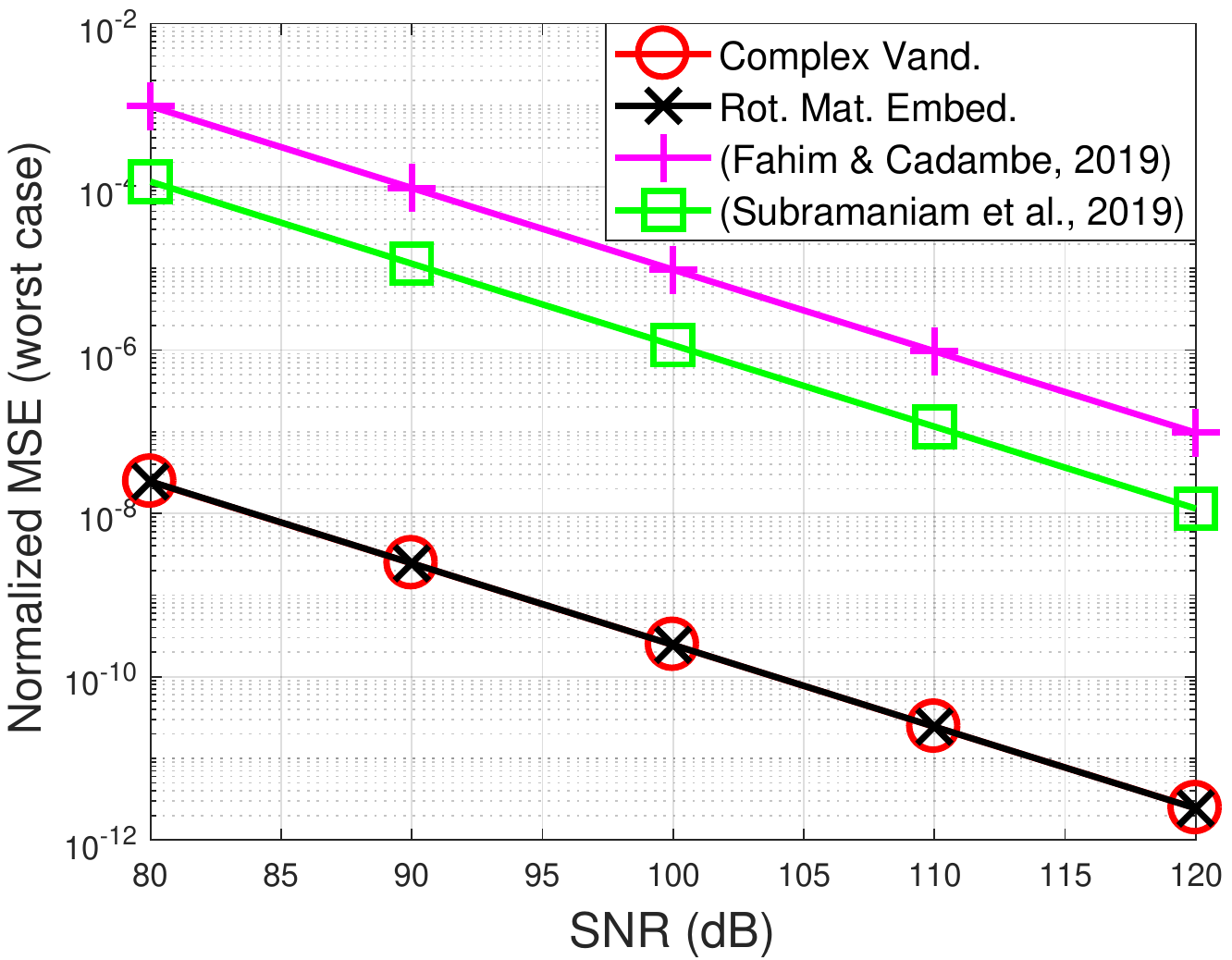}
	\vspace{-0.17in}
	\caption{{\small Consider matrix-matrix $\bfA^T\bfB$ multiplication system with $n=31$,  $k_A=4$, $k_B=7$, $\bfA$ is of size $8000\times 14000$, $\bfB$ is of $8400\times 14000$.}}
	\label{fig:mat_mat_mul_error_rate}
	\vspace{-0.2in}
\end{figure}

It can be observed that the decoding flop count for both matrix-vector and matrix-matrix multiplication is independent of $t$, i.e., in the regime where $t$ is very large the decoding time may be neglected with respect to the worker node computation time. Nevertheless, from a practical perspective it is useful to understand the decoding times as well.

When the matrix $\bfA$ is of dimension $28000 \times 19720$ and $\bfx$ is of length $28000$, the last two columns in Table \ref{table:mat_vac_cond} indicate the average worker node computation time and the master node decoding time for the different schemes. These numbers were obtained by averaging over several runs of the algorithm. It can be observed that the Complex Vandermonde scheme requires about twice the worker computation time as our schemes. Thus, it is wasteful of worker node computation resources. On the other hand, our schemes leverage the same condition number with computation over the reals. The decoding times of almost all the schemes are quite small. However, the Circulant Permutation Matrix scheme requires decoding time which is somewhat higher than the rotation matrix embedding even though we can use FFT based approaches for it. We expect that for much larger scale problems, the FFT based approach may be faster.

Our next set of results compare the mean-squared error (MSE) in the decoded result for the different schemes. To simulate numerical precision problems,
we added i.i.d Gaussian noise (of different SNRs) to the encoded submatrices of $\bfA$ and the vector $\bfx$ (the encoded submatrices of $\bfB$) in each worker node. The master node then performs decoding on the noisy vectors. The plots in Figure \ref{fig:mat_vec_mul_error_rate} correspond to the worst-case choice of worker nodes for each of the schemes. It can be observed that the Circulant Permutation Matrix Embedding has the best performance. This is because many of the matrices on the block-diagonal in (\ref{eq:diagonal_block_tilde_G_F}) (see Appendix \ref{sec:proof_circ_perm}) have well-behaved condition numbers and only a few correspond to the worst-case. We have not shown the results for the Real Vandermonde case here because the normalized MSE was very large.

\subsubsection{Matrix-Matrix case}
 In the matrix-matrix scenario we again consider a system with $n=31$ worker nodes and $k_A=4$ and $k_B=7$ so that the threshold $\tau = k_A k_B = 28$. Once again we observe ({\it cf.} Table \ref{table:mat_mat_cond}) that the worst-case condition number of the Rotation Matrix Embedding is about eleven orders of magnitude lower than the Real Vandermonde case.  Furthermore, the schemes of \cite{FahimC19} and \cite{subramaniam2019random} have a worst-case condition numbers that are two orders of magnitude higher than our scheme. For both \cite{subramaniam2019random} and \cite{DasRV19} schemes we performed $200$ random trials and picked the scheme with the lowest worst-case condition number. For \cite{DasRV19}, we only report the  upper bound on the worst-case condition number. Finding the actual worst-case recovery set takes a long time.

When the matrix $\bfA$ is of dimension $8000 \times 14000$ and $\bfB$ is of dimension $8000 \times 14000$, the worker node computation times and decoding times are listed in Table \ref{table:mat_mat_cond}. As expected the Complex Vandermonde scheme takes much longer for the worker node computations, whereas the Rotation Matrix Embedding, \cite{FahimC19} and \cite{subramaniam2019random} take about the same time. The decoding times are also very similar. As shown in Figure \ref{fig:mat_mat_mul_error_rate}, the normalized MSE of our Rotation Matrix Embedding scheme is much about five orders of magnitude lower than the scheme of \cite{FahimC19}. The normalized MSE of the Real Vandermonde case is very large so we do not plot it. Since we did not determine the worst-case recovery set for \cite{DasRV19}, we have not included the data and corresponding curves for it.

\begin{table*}[t]
	\centering
	\caption{Comparison for matrix-matrix $\bfA^T\bfB$ multiplication case with $n=17$,  $u_A=2$, $u_B=2$, $p=2$, $\bfA$ is of size $4000\times 16000$, $\bfB$ is of $4000\times 16000$.}
	\label{table:mat_mat_gen}
	\begin{tabular}{|c|c|c|c|c|c|c|c|}
		\hline
		{\small Scheme} &$\gamma_A$&$\gamma_B$&$\tau$& {\small Avg. Cond. Num.} & {\small Max. Cond. Num.} & {\small Avg. Worker Comp. Time (s)} & {\small Dec. Time (s)} \\
		\hline
		Rot. Mat. Embed.&$1/4$&$1/4$&$15$& $7$ & $22$&$2.23$ & $0.69$\\
		\hline
		Ortho-Poly \cite{FahimC19} &$1/4$&$1/4$& $15$&$10^4$ & $2.7\times 10^5$ & $2.23$ & $0.18$\\
		\hline
	\end{tabular}
\end{table*}

\begin{figure}[t]
	\centering
	\includegraphics[width=90mm]{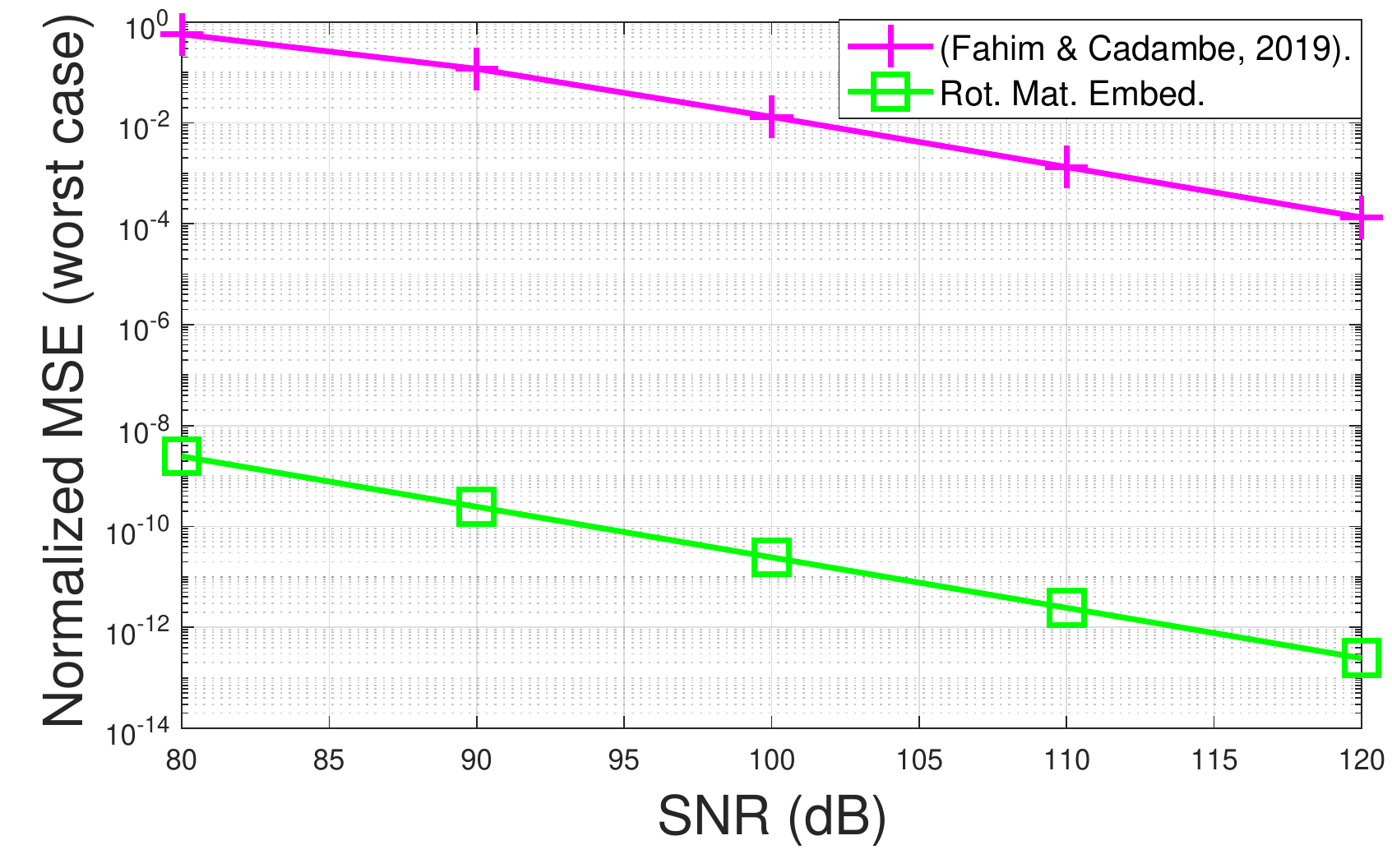}
	\caption{{\small Consider matrix-matrix $\bfA^T\bfB$ multiplication system with $n=18$, $u_A=2$, $u_B=2$, $p=2$, $\bfA$ is of size $4000\times 16000$, $\bfB$ is of $4000\times 16000$.}}
	\label{fig:mat_mat_mul_gen_error_rate}
\end{figure}

In the matrix-matrix multiplication scenario with $p\ge 2$, we consider a system with $n=17$ worker nodes and $u_A=2, u_B=2, p=2$. Note that in this case the threshold of \cite{yu2018straggler} is lower than our threshold and \cite{FahimC19}. Accordingly, we picked a setting where the our and \cite{FahimC19}'s threshold match and only compare these results.

We observe that the condition number of the Rotation Matrix Embedding scheme is about four orders of magnitude lower than \cite{FahimC19}. 
Figure \ref{fig:mat_mat_mul_gen_error_rate} shows that the normalized MSE of our Rotation Matrix Embedding scheme is much lower than \cite{FahimC19}. The Rotation Matrix Embedding scheme has higher decoding time since its decoding algorithm operates over the complex field.

%

\section{Conclusions and Future Work}
\label{sec:conclusions}
In this work we demonstrated that polynomial based schemes for coded computation suffer from serious numerical stability issues in practice. This stems from the provably bad conditioning of real Vandermonde matrices. We demonstrated a technique that exploits the properties of circulant and rotation matrices for coded computation. In essence, our method allows us to leverage the superior conditioning of complex Vandermonde matrices with parameters on the unit circle while still working with real computations at the worker nodes. The worst-case condition number of our recovery matrices is upper bounded by $O(n^{s+5.5})$ (where $n$- number of workers, $s$- number of stragglers) and our schemes have excellent performance in numerical experiments.

It is to be noted that our upper bound grows with the number of stragglers. In fact, it can be shown that if $s$ is a large fraction of $n$, then the condition number of the corresponding recovery matrices can be quite large even in the complex Vandermonde on unit circle case. It would be interesting to investigate coded computation schemes that continue to be numerically stable in the large $s$ regime.


\input{RamT21_arxiv_rev.bbl}


\newpage

\onecolumn


\appendix
\section{Appendix}

\subsection{Proof of Claim \ref{claim:fft_decoding_circ}}
\label{sec:proof_fft_dec}
\begin{proof}
	Note that Algorithm \ref{Alg:MatVecMul} is applied for recovering the corresponding entries of $\bfA_{i,j}^T \bfx$ for $i \in [k_A], j \in [\tilde{q}]$ separately. There are $r/(k_A(q-1))$ such entries.
	The complexity of computing a $N$-point FFT is $O(N \log N)$ in terms of the required floating point operations (flops). Computing the permutation does not cost any flops and its complexity is negligible as compared to the other steps. Step 1 of Algorithm \ref{Alg:MatVecMul} therefore has complexity $O(k_A \tilde{q} \log \tilde{q})$. In Step 2, we solve the degree $k_A-1$ polynomial interpolation, $(\tilde{q}-1)$ times. This takes $O((\tilde{q}-1) k_A \log^2 k_A)$ time \cite{Pan2013TR}. Finally, Step 3, requires applying the inverse permutation and the inverse FFT; this requires $O(k_A \tilde{q} \log \tilde{q})$ operations. Therefore, the overall complexity is given by
	\begin{align*}
		&\frac{r}{k_A(\tilde{q}-1)} \left( O(k_A \tilde{q} \log \tilde{q})  + O((\tilde{q}-1) k_A \log k_A^2) \right)\\
		&\approx O(r (\log \tilde{q} + \log^2 k_A)).
	\end{align*}
\end{proof}

\subsection{Proof of Theorem \ref{thm:circ_perm_scheme}}
\label{sec:proof_circ_perm}
\begin{proof}
	The arguments are conceptually similar to the proof of Theorem \ref{thm:rotmat_mv}. Suppose that the workers indexed by $i_0, \dots, i_{k_A -1}$ complete their tasks. The corresponding block-columns of $\bfG^{circ}$ can be extracted to form
	\begin{align*}
		\tilde{\bfG} =
		\begin{bmatrix}
			\bfI & \bfI & \cdots & \bfI\\
			\bfP^{i_0} &  \bfP^{i_1} &\cdots &\bfP^{i_{k_A-1}}\\
			\vdots & \vdots &\ddots &\vdots\\
			\bfP^{i_0(k_A-1)} & \bfP^{i_1(k_A-1)} & \cdots & \bfP^{i_{k_A-1}(k_A-1)}
		\end{bmatrix}.
	\end{align*}

As in the proof of Theorem \ref{thm:rotmat_mv} we can equivalently analyze the decoding by considering the system of equations 	
\begin{align*}
		\bfm \tilde{\bfG} &= \bfc,
\end{align*}
where $\bfm, \bfc \in \mathbb{R}^{1 \times k_A \tilde{q}}$  are row-vectors such that
	\begin{align*}
		\bfm&=[\bfm_{0},\cdots, \bfm_{k_A-1}]\\
		 &= [\bfm_{\langle 0,0 \rangle}, \cdots, \bfm_{\langle 0,\tilde{q}-1 \rangle}, \cdots, \cdots \bfm_{\langle k_A-1, 0 \rangle}, \cdots, \bfm_{\langle k_A-1, \tilde{q}-1 \rangle}], \text{~and} \\
		\bfc&=[\bfc_{i_0}, \cdots, \bfc_{i_{k_A-1}}] \\&= [\bfc_{\langle i_0,0 \rangle},\cdots, \bfc_{\langle i_0,\tilde{q}-1 \rangle} , \cdots, \cdots, \bfc_{\langle i_{k_A-1},0 \rangle}, \cdots, \bfc_{\langle i_{k_A-1}, \tilde{q}-1 \rangle}].
	\end{align*}
Note that not all variables in $\bfm$ are independent owing to (\ref{eq:precoded_mat}). Let $\bfm^\calF$ and $\bfc^\calF$ denote the $\tilde{q}$-point ``block-Fourier'' transforms of these vectors, i.e,
	\begin{align*}
		\bfm^\calF &= \bfm \begin{bmatrix}
			\bfW&&\\
			&\ddots&\\
			&&\bfW
		\end{bmatrix}\text{~and}\\
		\bfc^\calF &= \bfc \begin{bmatrix}
			\bfW&&\\
			&\ddots&\\
			&&\bfW
		\end{bmatrix},
	\end{align*}
where $\bfW$ is the $\tilde{q}$-point DFT matrix. Let $\tilde{\bfG}_{k,l} = \bfP^{i_l k}$ denote the $(k,l)$-th block of $\tilde{\bfG}$. Using the fact that $\bfP$ can be diagonalized by the DFT matrix $\bfW$, we have
	\begin{align*}
		\tilde{\bfG}_{k,l}  &= \bfW \text{diag}(1, \omega_{\tilde{q}}^{i_lk}, \omega_{\tilde{q}}^{2i_l k}, \dots, \omega_{\tilde{q}}^{(\tilde{q}-1)i_l k}) \bfW^{*}.
	\end{align*}
	Let $\tilde{\bfG}_{k,l}^\calF = \text{diag}(1, \omega_{\tilde{q}}^{i_l k}, \omega_{\tilde{q}}^{2i_l k}, \dots, \omega_{\tilde{q}}^{(\tilde{q}-1)i_l k})$, and $\tilde{\bfG}^\calF$ represent the $k_A \times k_A$ block matrix with $\tilde{\bfG}_{k,l}^\calF$ for $k,l = 0, \dots, k_A -1$ as its blocks. Therefore, the system of equations
	\begin{align*}
		\bfm \tilde{\bfG} &= \bfc,
	\end{align*}
can be further expressed as
	\begin{align*}
		\bfm \begin{bmatrix}
			\bfW&&\\
			&\ddots&\\
			&&\bfW
		\end{bmatrix}
        \begin{bmatrix}
			\bfW^*&&\\
			&\ddots&\\
			&&\bfW^*
		\end{bmatrix}
		\tilde{\bfG} \begin{bmatrix}
			\bfW&&\\
			&\ddots&\\
			&&\bfW
		\end{bmatrix}
		 &= \bfc \begin{bmatrix}
			\bfW&&\\
			&\ddots&\\
			&&\bfW
		\end{bmatrix},\\
		\implies [\bfm_{0}^\calF,\cdots, \bfm_{k_A-1}^\calF] \tilde{\bfG}^\calF &= [\bfc_{i_0}^\calF, \cdots, \bfc_{i_{k_A-1}}^\calF]
	\end{align*}
upon right multiplication by the matrix $\begin{bmatrix}
			\bfW&&\\
			&\ddots&\\
			&&\bfW
		\end{bmatrix}$. Next, we note that as each block within $\tilde{\bfG}^\calF$ has a diagonal structure, we can rewrite the system of equations as a block diagonal matrix upon applying an appropriate permutation ({\it cf.} Claim \ref{claim:diag_block_matrix} in Appendix \ref{sec:aux_claims}).  Thus, we can rewrite it as
	\begin{align}
		[\bfm_{0}^{\calF,\pi},\cdots, \bfm_{\tilde{q}-1}^{\calF,\pi}] \tilde{\bfG}^\calF_d &= [\bfc_{0}^{\calF,\pi}, \cdots, \bfc_{\tilde{q}-1}^{\calF,\pi}], \label{eq:permuted_eq_mat_vec}
	\end{align}
	where the permutation $\pi$ is such that $\bfm_{j}^{\calF,\pi} = [\bfm_{0,j}^\calF ~ \bfm_{1,j}^\calF ~ \dots ~ \bfm_{k_A-1,j}^\calF]$ and likewise $\bfc_{j}^{\calF,\pi} = [\bfc_{i_0,j}^\calF ~ \bfc_{i_1,j}^\calF ~ \dots ~ \bfc_{i_{k_A-1},j}^\calF]$. Furthermore, $\tilde{\bfG}^\calF_d$ is a block-diagonal matrix where each block is of size $k_A \times k_A$. Now, according to (\ref{eq:precoded_mat}), we have $\bfm_{i,0}^\calF = \sum_{j = 0}^{\tilde{q}-1} \bfm_{i,j}=0$ for $i = 0, \dots, k_A-1$, which implies that $\bfm_{0}^{\calF,\pi}$ is a $1 \times k_A$ zero row-vector and thus $\bfc_{0}^{\calF,\pi}$ is too.
	
	In what follows, we show that each of the other diagonal blocks of $\tilde{\bfG}^\calF_d$ is non-singular. This means that $[\bfm_{0}^\calF,\cdots, \bfm_{k_A-1}^\calF]$ and consequently $\bfm$ can be determined by solving the system of equations in (\ref{eq:permuted_eq_mat_vec}). Towards this end, we note that the $k$-th diagonal block $(1 \leq k \leq \tilde{q}-1)$ of $\tilde{\bfG}^\calF_d$, denoted by $\tilde{\bfG}^\calF_d[k]$ can be expressed as follows.
	
	\begin{align}
		\label{eq:diagonal_block_tilde_G_F}
		\tilde{\bfG}^\calF_d[k] = \begin{bmatrix}
			1 & 1 &  \cdots & 1\\
			\omega_{\tilde{q}}^{i_0 k} & \omega_{\tilde{q}}^{i_1 k} & \cdots & \omega_{\tilde{q}}^{i_{k_A-1}k}\\
			\vdots & \vdots & \ddots & \vdots\\
			\omega_{\tilde{q}}^{(k_A-1)i_0k} & \omega_{\tilde{q}}^{(k_A-1)i_1k}  & \cdots & \omega_{\tilde{q}}^{(k_A-1)i_{k_A-1}k}
		\end{bmatrix}.
	\end{align}
	The above matrix is a complex Vandermonde matrix with parameters $\omega_{\tilde{q}}^{i_0 k}, \dots, \omega_{\tilde{q}}^{i_{k_A-1} k}$. Thus, as long these parameters are distinct, $\tilde{\bfG}^\calF_d[k]$ will be non-singular. Note that we need the property to hold for $k = 1, \dots, \tilde{q}-1$. This condition can be expressed as
	\begin{align*}
		(i_\alpha - i_\beta) k \not\equiv 0 \pmod{\tilde{q}},
	\end{align*}
	for $i_\alpha, i_\beta \in \{0, \dots, n-1\}$ and $1 \leq k \leq \tilde{q}-1$. A  necessary and sufficient condition for this to hold is that $\tilde{q}$ is prime.
	An application of Theorem \ref{thm:cond_no_vand} shows that $\kappa(\tilde{\bfG}^\calF_d[k]) \leq O(\tilde{q}^{\tilde{q}-k_A + c_1})$ for all $k$. As decoding $\bfm$ is equivalent to solving systems of equations specified by $\tilde{\bfG}^\calF_d[k]$ for $1 \leq k \leq \tilde{q}-1$, the worst-case condition number is at most $O(\tilde{q}^{\tilde{q}-k_A+c_1})$.
\end{proof}

\subsection{Vandermonde Matrix condition number analysis}
\label{sec:proof_vand_cond_no}
Let $\bfV$ be a $m \times m$ Vandermonde matrix with parameters $s_0, s_1, \dots s_{m-1}$. We are interested in upper bounding $\kappa(\bfV)$.  Let $s_{+} = \max_{i=0}^{m-1} |s_i|$. Then, it is known that $||\bfV|| \leq m \max (1, s_{+}^{m-1})$ \cite{Pan16}. Finding an upper bound on $||\bfV^{-1}||$ is more complicated and we discuss this in detail below. Towards this end we need the definition of a Cauchy matrix.
\begin{definition}
	A $m \times m$ Cauchy matrix is specified by parameters $\bfs = [s_0 ~s_1~\dots~s_{m-1}]$ and $\bft = [t_0 ~t_1~\dots~t_{m-1}]$, such that its $(i,j)$-th entry
	\begin{align*}
	\bfC_{\bfs,\bft}(i,j) &= \left( \frac{1}{s_i - t_j}\right) \text{~for $i \in [m], j \in [m]$.}
	\end{align*}

\end{definition}

In what follows, we establish an upper bound on the condition number of Vandermonde matrices with parameters on the unit circle.

\noindent {\it Proof of Theorem \ref{thm:cond_no_vand}}
\begin{proof}
	Recall that $\omega_q = e^{\rm{i}\frac{2\pi}{q}}$ and $\omega_m = e^{\rm{i}\frac{2\pi}{m}}$ and define $t_j = f \omega_m^j, j = 0, \dots, m-1$ where $f$ is a complex number with $|f| = 1$.  We let $\bfC_{\bfs,f}$ denote the Cauchy matrix with parameters $\{s_0, \dots, s_{m-1}\}$ and $\{t_0, \dots, t_{m-1}\}$. Let $\bfW$ be the $m$-point DFT matrix. The work of \cite{Pan16} shows that 
	\begin{align*}
	\bfV^{-1} &= \text{diag}(f^{m-1-j})_{j=0}^{m-1} \sqrt{m}\bfW^* \text{diag}(\omega_m^{-j})_{j=0}^{m-1} \bfC_{\bfs,f}^{-1} \text{diag} \left(\frac{1}{s_j^m - f^m}\right)_{j=0}^{m-1}.
	\end{align*}
	It can be seen that the matrix $\text{diag}(f^{m-1-j})_{j=0}^{m-1} \bfW^* \text{diag}(\omega_m^{-j})_{j=0}^{m-1}$ is unitary. Therefore,
	\begin{align}
	||\bfV^{-1}|| &=  \sqrt{m} ||\bfC_{\bfs,f}^{-1} \text{diag} \left(\frac{1}{s_j^m - f^m}\right)_{j=0}^{m-1}|| \nonumber\\
	&\leq \sqrt{m} ||\bfC_{\bfs,f}^{-1}|| \times \left( \frac{1}{\min_{i=0}^{m-1} |s_i^m - f^m|} \right) \nonumber\\
    &\leq m^{1.5} \times (\max_{i',j'} |\bfC_{\bfs,f}^{-1}(i',j')|) \times \left( \frac{1}{\min_{i=0}^{m-1} |s_i^m - f^m|} \right), \label{eq:upper_bound_cauchy}
	\end{align}
	where the first inequality holds as the norm of a product of matrices is upper bounded by the products of the individual norms and second inequality holds since for any $\bfM$, we have $||\bfM|| \leq ||\bfM||_F$ .
	
	In what follows, we upper bound the RHS of (\ref{eq:upper_bound_cauchy}). Let $s(x)$ denote a function of $x$ so that $s(x) = \Pi_{i=0}^{m-1} (x - s_i)$. The $(i',j')$-the entry of $\bfC_{\bfs,f}^{-1}$ can be expressed as \cite{Pan16}
	\begin{align*}
    \bfC_{\bfs,f}^{-1}(i',j') &= (-1)^m s(t_{j'})(s_{i'}^m - f^m)/(s_{i'} - t_{j'}), \text{~so that}\\
    |\bfC_{\bfs,f}^{-1}(i',j')| &=  |s(t_{j'})| |s_{i'}^m - f^m|/|s_{i'} - t_{j'}| \\
	&\leq |s(t_{j'})| (|s_{i'}^m| + |f^m|)/|s_{i'} - t_{j'}| \\
	&= 2 |s(t_{j'})|/|s_{i'} - t_{j'}| \text{~~(since $|s_{i'}| = |f| = 1$)}.
	\end{align*}
	Let $\calM = \{1, \omega_q, \omega_q^2, \dots, \omega_q^{q-1}\} \setminus \{s_0, s_1, \dots, s_{m-1}\}$ denote the $q$-th roots of unity that are {\it not} parameters of $\bfV$. Note that
	\begin{align*}
	s(t_{j'}) & = \Pi_{i=0}^{m-1} (t_{j'} - s_i)\\
	&= \frac{x^q-1}{\Pi_{\alpha_j \in \calM} (x - \alpha_j)}\bigg{|}_{x = t_{j'}}, \text{~so that}\\
	|s(t_{j'})| &= \frac{|t^q_{j'} - 1|}{\Pi_{\alpha_j \in \calM} |t_{j'} - \alpha_j|}\\
	&\leq \frac{2}{\Pi_{\alpha_j \in \calM} |t_{j'} - \alpha_j|} \text{~(since $|t_{j'}| = 1$ and by the triangle inequality).}
	\end{align*}
	Thus, we can conclude that
	\begin{align}
    \max_{i',j'} |\bfC_{\bfs,f}^{-1}(i',j')| &\leq 4 \max_{i',j'} \frac{1}{\Pi_{\alpha_j \in \calM } |(t_{j'} - \alpha_j)|} \frac{1}{|s_{i'} - t_{j'}|}\\
	&= 4 \left(\frac{1}{\min_{i',j'} \Pi_{\alpha_j \in \calM } |(t_{j'} - \alpha_j)|} \frac{1}{|s_{i'} - t_{j'}|} \right). \label{eq:tmp_bd}
	\end{align}
	Note that in the expression above the $\alpha_j$'s and $s_{i'}$ are all points within $\Omega_q = \{1, \omega_q, \omega_q^2, \dots, \omega_q^{q-1}\}$.
	We choose $f=e^{{\rm{i}} \frac{\pi}{m}}$ so that $t_{j'} = f \omega_m^{j'} = e^{{\rm{i}} \frac{\pi}{m}} \omega_m^{j'}$. Now for any $i'$ and $j'$ we need to lower bound $\Pi_{\alpha_j \in \calM } |(t_{j'} - \alpha_j)| |s_{i'} - t_{j'}|$. Towards this end, we note that the distance between two points on the unit circle can be expressed as $2 \sin (\theta/2)$ if $\theta$ is the induced angle between them. Furthermore, we have $2 \sin(\theta/2) \geq 2 \theta/\pi$ as long as $\theta \leq \pi$.
	
	\noindent Let $d = q-m$. Then, for any choice of $t_{j'}$ we can consider lower bounds on the distances of $d+1$ points that lie on $\Omega_q$. It can be seen that the closest point to $t_{j'}$ that lies within $\Omega_q$ has an induced angle 
	\begin{align*}
	\bigg{|}\frac{2\pi \ell}{q} - \frac{2\pi (j' + \frac{1}{2})}{m}\bigg{|} \geq \frac{2\pi}{q m} \frac{1}{2} \geq \frac{\pi}{q^2} \text{~(since $q$ is odd and $q > m$).}
	\end{align*}
	Therefore, the corresponding distance is lower bounded by $2/q^2$.
	Similarly, the next closest distance is lower bounded by $2/q$, followed by $2(2/q), 3(2/q), \dots, d(2/q)$. Then, 
	\begin{align*}
	& \min_{i',j'} \left(\Pi_{\alpha_j \in \calM } |(t_{j'} - \alpha_j)| \right)  |s_{i'} - t_{j'}| \\
	&\geq 2/q^2 \times 2/q \times 4/q \times \dots \times 2d/q\\
	&=2^{d+1}d! \frac{1}{q^{d+2}}.
	\end{align*}
	Therefore,
	\begin{align*}
	\max_{i',j'} |\bfC_{\bfs,f}^{-1}(i',j')| &\leq \frac{q^{d+2}}{C_d}
	\end{align*}
	where $C_d = 2^{d-1}d!$ is a constant.
	Let the $i$-th parameter $s_i= e^{\rm{i} 2\pi \ell/q}$. Then,
	\begin{align*}
	|s_i^m - f^m| &= |e^{\rm{i} 2\pi \ell m/q} + 1|\\
	&= 2 |\cos (\pi \ell m/q)|.
	\end{align*}
	The term $\ell m$ can be expressed as $\ell m = \beta q + \eta$ for integers $\beta$ and $\eta$ such that $0 \leq \eta \leq q-1$. Now note that $\eta \neq q/2$ since by assumption $q$ is odd. Thus, $|\cos (\pi \ell m/q)|$ takes its smallest value when $\eta = (q+1)/2$ or $(q-1)/2$. In this case
	\begin{align*}
	|\cos (\pi \ell m/q)| &= \bigg{|}\cos\left(\beta \pi + \pi\frac{q+1}{2q}\right)\bigg{|}\\
	&\geq \bigg{|}\sin \left(\frac{\pi}{2q}\right)\bigg{|}\\
	&\geq \frac{1}{q}.
	\end{align*}
	
	Thus, we can upper bound the RHS of (\ref{eq:upper_bound_cauchy}) and obtain
	\begin{align*}
	||\bfV^{-1}|| &\leq m^{1.5} \frac{q^{d+2}}{C_d} q\\
	&\leq \frac{q^{d+4.5}}{C_d} \text{~(since $m < q$).}
	\end{align*}
	Finally, using the fact that $||V|| \leq m < q$. we obtain
	\begin{align*}
	\kappa(\bfV) &\leq \frac{q^{d+5.5}}{C_d}.
	\end{align*}
\end{proof}


\subsection{Proof of Theorem \ref{thm:genMatMat}}
\label{sec:proof_genMatMat}
\begin{proof}
	We proceed in a similar manner as in Example \ref{eg:gen_matmat}.
     Following the encoding rules ({\it cf.} Algorithm \ref{Alg:Gen_MM_Embedding_Scheme}) and worker computation rules ({\it cf.} (\ref{eq:gen_matmat_worker_comp})), we can analyze the computation in worker $k$ as follows. Let $(\bfQ^*\otimes \bfI_{\zeta})\begin{bmatrix}
	\bfA_{(\langle i, 0 \rangle, j)}\\
	\bfA_{(\langle i, 1 \rangle, j)}
	\end{bmatrix}=\begin{bmatrix}
	\bfA^\calF_{(\langle i, 0 \rangle, j)}\\
	\bfA^\calF_{(\langle i, 1 \rangle, j)}
	\end{bmatrix}$
	and
	$(\bfQ^* \otimes \bfI_{\zeta})\begin{bmatrix}
	\bfB_{(\langle i, 0 \rangle, j)}\\
	\bfB_{(\langle i, 1 \rangle, j)}
	\end{bmatrix}=\begin{bmatrix}
	\bfB^\calF_{(\langle i, 0 \rangle, j)}\\
	\bfB^\calF_{(\langle i, 1 \rangle, j)}
	\end{bmatrix}$.
	Let $\hat{\bfA}_k = \begin{bmatrix}
	\hat{\bfA}_{\langle k, 0 \rangle}\\
	\hat{\bfA}_{\langle k, 1 \rangle}
	\end{bmatrix}$ and  $\hat{\bfB}_k = \begin{bmatrix}
	\hat{\bfB}_{\langle k, 0 \rangle}\\
	\hat{\bfB}_{\langle k, 1 \rangle}
	\end{bmatrix}$.
     Then, we have
	\begin{align*}
	\hat{\bfA}_k^\calF &=  (\bfQ^*\otimes \bfI_{\zeta}) \hat{\bfA}_k = \sum_{i=0}^{p-1}\sum_{j=0}^{k_A-1} (\bfQ^* \bfR_{-\theta}^{k((j-1)p+i+1)}\bfQ \bfQ^* \otimes \bfI_{\zeta})\begin{bmatrix}
	\bfA_{(\langle i, 0 \rangle, j)} \\ \bfA_{(\langle i, 1 \rangle, j)}
	\end{bmatrix}\\
	& = \sum_{i=0}^{p-1}\sum_{j=0}^{k_A-1} ({\Lambda^*}^{k((j-1)p+i+1)}\otimes \bfI_{\zeta}) (\bfQ^* \otimes \bfI_{\zeta})\begin{bmatrix}
	\bfA_{(\langle i, 0 \rangle, j)} \\ \bfA_{(\langle i, 1 \rangle, j)}
	\end{bmatrix}\\
	&=\begin{bmatrix}
	\sum_{i=0}^{p-1}\sum_{j=0}^{k_A-1}{\omega_q^*}^{k((j-1)p+i+1)}\bfA_{(\langle i, 0 \rangle, j)}^\calF\\
	\sum_{i=0}^{p-1}\sum_{j=0}^{k_A-1}{\omega_q^*}^{-k((j-1)p+i+1)}\bfA_{(\langle i, 1 \rangle, j)}^\calF
	\end{bmatrix}, \text{~and}\\
	\hat{\bfB}_k^\calF &=  (\bfQ^* \otimes \bfI_{\zeta}) \hat{\bfB}_k = \sum_{i=0}^{p-1}\sum_{j=0}^{k_B-1} (\bfQ^* \bfR_{\theta}^{k(p-1-i+jpk_A)}\bfQ \bfQ^* \otimes \bfI_{\zeta}) \begin{bmatrix}
	\bfB_{(\langle i, 0 \rangle, j)} \\ \bfB_{(\langle i, 1 \rangle, j)}
	\end{bmatrix}\\
	&=\sum_{i=0}^{p-1}\sum_{j=0}^{k_B-1}(\Lambda^{k(p-1-i+jpk_A)}\otimes\bfI_{\zeta})(\bfQ^* \otimes \bfI_{\zeta})\begin{bmatrix}
	\bfB_{(\langle i, 0 \rangle, j)}^\calF \\ \bfB_{(\langle i, 1 \rangle, j)}^\calF
	\end{bmatrix}\\
	&=\begin{bmatrix}
	\sum_{i=0}^{p-1}\sum_{j=0}^{k_B-1}\omega_q^{k(p-1-i+jpk_A)}\bfB_{(\langle i, 0 \rangle, j)}^\calF\\
	\sum_{i=0}^{p-1}\sum_{j=0}^{k_A-1}\omega_q^{-k(p-1-i+jpk_A)}\bfB_{(\langle i, 1 \rangle, j)}^\calF
	\end{bmatrix}.
	\end{align*}
    This implies that
	\begin{align}
	\label{eq:generalize_sum}
	\begin{split}
	\hat{\bfA}_{k}^T 	\hat{\bfB}_{k} =&((\bfQ^*\otimes \bfI_{\zeta}) \hat{\bfA}_k)^* (\bfQ^*\otimes \bfI_{\zeta}) \hat{\bfB}_k\\
	=& \hat{\bfA}_k^{\calF *}  \hat{\bfB}_k^{\calF}\\
	=& \bigg(\sum_{i=0}^{p-1}\sum_{j=0}^{k_A-1}\omega_q^{k((j-1)p+i+1)}\bfA_{(\langle i, 0 \rangle, j)}^{\calF *}\bigg) \bigg(\sum_{i=0}^{p-1}\sum_{j=0}^{k_B-1}\omega_q^{k(p-1-i+jpk_A)}\bfB_{(\langle i, 0 \rangle, j)}^\calF\bigg)+\\
	&\bigg(\sum_{i=0}^{p-1}\sum_{j=0}^{k_A-1}\omega_q^{-k((j-1)p+i+1)}\bfA_{(\langle i, 1 \rangle, j)}^{\calF *}\bigg) \bigg(\sum_{i=0}^{p-1}\sum_{j=0}^{k_B-1}\omega_q^{-k(p-1-i+jpk_A)}\bfB_{(\langle i, 1 \rangle, j)}^\calF\bigg).
	\end{split}	
	\end{align}
	
	To better understand the behavior of the sum in (\ref{eq:generalize_sum}), we divide it into the following two cases.
	\begin{itemize}
		\item {\it Case 1: Useful terms.} The master node wants to recover $\bfC=\bfA^T\bfB=[\bfC_{i,j}]$, $i\in [k_A]$, $j\in [k_B]$, where each $\bfC_{i,j}$ is a block matrix of size $r/k_A \times w/k_B$. 
Note that $\bfC_{i,j}=\sum_{u=0}^{p-1}(\bfA_{(\langle u, 0 \rangle, i)}^T \bfB_{(\langle u, 0 \rangle, j)} +  \bfA_{(\langle u, 1 \rangle, i)}^T \bfB_{(\langle u, 1 \rangle, j)})$.

Moreover, note that
		\begin{align*}
		&\bfA_{(\langle u, 0 \rangle, i)}^{\calF *} \bfB_{(\langle u, 0 \rangle, j)}^{\calF} + \bfA_{(\langle u, 1 \rangle, i)}^{\calF *} \bfB_{(\langle u, 1 \rangle, j)}^{\calF}\\
		=&\begin{bmatrix}
		\bfA^\calF_{(\langle u, 0 \rangle, i)}\\
		\bfA^\calF_{(\langle u, 1 \rangle, i)}
		\end{bmatrix}^*
		\begin{bmatrix}
		\bfB^\calF_{(\langle u, 0 \rangle, j)}\\
		\bfB^\calF_{(\langle u, 1 \rangle, j)}
		\end{bmatrix}\\
		=&\bigg((\bfQ^* \otimes \bfI_{\zeta})\begin{bmatrix}
		\bfA_{(\langle u, 0 \rangle, i)}\\
		\bfA_{(\langle u, 1 \rangle, i)}
		\end{bmatrix}\bigg)^*
		\bigg((\bfQ^* \otimes \bfI_{\zeta}) \begin{bmatrix}
		\bfB_{(\langle u, 0 \rangle, j)}\\
		\bfB_{(\langle u, 1 \rangle, j)}
		\end{bmatrix}\bigg)\\
		=&\begin{bmatrix}
		\bfA_{(\langle u, 0 \rangle, i)}\\
		\bfA_{(\langle u, 1 \rangle, i)}
		\end{bmatrix}^*
		\begin{bmatrix}
		\bfB_{(\langle u, 0 \rangle, j)}\\
		\bfB_{(\langle u, 1 \rangle, j)}
		\end{bmatrix}\\
		=&\bfA_{(\langle u, 0 \rangle, i)}^T \bfB_{(\langle u, 0 \rangle, j)} + \bfA_{(\langle u, 1 \rangle, i)}^T \bfB_{(\langle u, 1 \rangle, j)}.
		\end{align*}
		It is easy to check that $\sum_{u=0}^{p-1} \bfA_{(\langle u, 0 \rangle, i)}^{\calF *} \bfB_{(\langle u, 0 \rangle, j)}^{\calF}$ is the coefficient of $\omega_q^{k(ip+jpk_A)}$ and $\sum_{u=0}^{p-1} \bfA_{(\langle u, 1 \rangle, i)}^{\calF *} \bfB_{(\langle u, 1 \rangle, j)}^{\calF}$ is the coefficient of $\omega_q^{-k(ip+jpk_A)}$. Thus, decoding and summing the corresponding coefficients, allows us to recover $\bfC_{i,j}$. Note further that the exponent of $\omega_q$ is a multiple of $p$.
		\item {\it Case 2: Interference terms.} The terms in (\ref{eq:generalize_sum}) with coefficient $\bfA_{(\langle u, l \rangle, i)}^{\calF *} \bfB_{(\langle v, l \rangle, j)}^{\calF}$ with $u\neq v$ are the \emph{interference terms} and they are the coefficients of $\omega_q^{\pm k(ip+u-v+jpk_A)}$. We conclude that the useful terms have no intersection with interference terms since $1 \leq |u-v|<p$. 
	\end{itemize}
	
	Next we determine the threshold of the proposed scheme. Towards this end, we find the maximum and minimum degree of $\hat{\bfA}_k^{\calF *}  \hat{\bfB}_k^{\calF}$ and then argue that (\ref{eq:generalize_sum}) has powers of $\omega_q$ that lie at consecutive multiples of $k$. The threshold can then be obtained by adding 1 to the difference of the maximum and minimum degrees divided by $k$. The maximum degree of $\hat{\bfA}_k^{\calF *}  \hat{\bfB}_k^{\calF}$ is the degree of the term
	\begin{align*}
	\omega_q^{k(pk_Ak_B-1)}\bfA_{(\langle p-1,0\rangle, k_A-1)}^{\calF *}\bfB_{(\langle 0,0\rangle, k_B-1)}^{\calF},
	\end{align*}
	and the minimum degree is the degree of the term
	\begin{align*}
	\omega_q^{-k(pk_Ak_B-1)}\bfA_{(\langle p-1,1\rangle, k_A-1)}^{\calF *} \hat{\bfB}_{(\langle 0,1\rangle, k_B-1)}^{\calF}.
	\end{align*}
	
	Next we argue that (\ref{eq:generalize_sum}) has powers of $\omega_q$ that are consecutive multiples of $k$ between the  maximum and minimum degree. Towards this end, we show that there always exist some terms in (\ref{eq:generalize_sum}) with degree $dk$, where $-pk_Ak_B+1\le d \le pk_Ak_B-1$.   We observe that the positive powers of ${\omega_q}^k$ in (\ref{eq:generalize_sum}) can be written as
	$\pm ((j_1-1)p+i_1+1+p-1-i_2+j_2pk_A)=\pm (j_2pk_A+j_1p+i_1-i_2)$, where $j_1\in [k_A], j_2\in [k_B], i_1, i_2\in [p]$. Consider a positive power $d\le pk_Ak_B-1$. We can always find a solution such that $j_2=\lfloor \frac{d}{pk_A}\rfloor$, $j_1= \lfloor \frac{d \text{~mod~}pk_A}{p} \rfloor$, $i_1-i_2=(d \text{~mod~}pk_A) \text{~mod~} p$. A similar result holds when $d$ is negative. We conclude that the threshold of the scheme is $2pk_Ak_B-1$. 
	
	Now suppose that $2pk_Ak_B-1$ workers return their results. Equation (\ref{eq:generalize_sum}) shows that the condition number of the corresponding decoding matrix is equivalent to (up to multiplication by an appropriately defined unitary matrix) a Vandermonde matrix whose parameters are a $(2pk_Ak_B-1)$- sized subset of $\{1, \omega_q, \omega_q^2, \dots, \omega_q^{q-1}\}$.
	Therefore, an application of Theorem \ref{thm:cond_no_vand} implies that the worst-case condition number is upper bounded by $O(q^{q-2pk_Ak_B+1+c_1})$.
	
\end{proof}

\subsection{Auxiliary Claims}
\label{sec:aux_claims}

\begin{definition}{\it Permutation Equivalence.}
	We say that a matrix $\bfM$ is \emph{permutation equivalent} to $\bfM^\pi$ if $\bfM^\pi$ can be obtained by permuting the rows and columns of $\bfM$. We denote this by $\bfM \asymp \bfM^\pi$.
\end{definition}

\begin{claim}
	\label{claim:diag_block_matrix}
	Let $\bfM$  be a $l_1 q \times l_2 q$ matrix consisting of blocks of size $q \times q$ denoted by $\bfM_{i,j}$ for $i \in [l_1], j \in [l_2]$ where each $\bfM_{i,j}$ is a diagonal matrix.
	Then, the rows and columns of $\bfM$ can be permuted to obtain $\bfM^{\pi}$ which is a block diagonal matrix where each block matrix is of size $l_1 \times l_2$ and there are $q$ of them.
\end{claim}
\begin{proof}
	For an integer $a$, let $(a)_{q}$ denote $a \bmod q$. In what follows, we establish two permutations
	\begin{align*}
	\pi_{l_1}(i)&=l_1(i)_{q}+\lfloor i/q\rfloor, 0\le i<l_1q, \text{~and}\\
	\pi_{l_2}(j)&=l_2(j)_{q}+\lfloor j/q\rfloor, 0\le j<l_2q
	\end{align*}
	and show that applying row-permutation $\pi_{l_1}$ and column-permutation $\pi_{l_2}$ to $\bfM$ will result in a block diagonal matrix $\bfM^{\pi}$.
	
	We observe that $(i, j)$-th entry in $\bfM$ is the $((i)_q, (j)_q)$-th entry in the block $\bfM_{\lfloor i/q\rfloor, \lfloor j/q\rfloor}$. Under the applied permutations the $(i,j)$-th entry in $\bfM$ is mapped to $(l_1(i)_q+\lfloor i/q\rfloor, l_2(j)_q+\lfloor j/q\rfloor)$-entry in $\bfM^\pi$. Recall that $\bfM_{\lfloor i/q\rfloor, \lfloor j/q\rfloor}$ is a diagonal matrix which implies that for $(i)_q\neq (j)_q$, the $(l_1(i)_q+\lfloor i/q\rfloor, l_2(j)_q+\lfloor j/q\rfloor)$ entry in $\bfM^{\pi}$ is $0$. Therefore $\bfM^\pi$ is a block diagonal matrix with $q$ blocks of size $l_1 \times l_2$.
\end{proof}

\begin{example}
	Let $l_1=2, l_2=3, q=2$. Consider a $4\times 6$ matrix $\bfM$ which consists of diagonal matrices $\bfM_{i,j}$ of size $2\times 2$. For $0 \leq i \leq 1, 0 \leq j \leq 2$
	\begin{align*}
	\bfM &= \begin{bmatrix}
	\bfM_{0,0} & \bfM_{0,1} & \bfM_{0,2}\\
	\bfM_{1,0} & \bfM_{1,1} & \bfM_{1,2}
	\end{bmatrix}\\
	&=\begin{bmatrix}
	1 & 0 & 1 & 0 & 1 & 0\\
	0 & 1 & 0 & 1 & 0 & 1\\
	1 & 0 & \omega_q & 0 & \omega_q^2 & 0\\
	0 & 1 & 0 & \omega_q^{-1} & 0 & \omega_q^{-2}
	\end{bmatrix}.
	\end{align*}
	
	We use row permutation $\pi_{\text{row}} = (0,2,1,3)$, which means $0, 1, 2, 3$-th row of $\bfM$ permutes to $0, 2, 1, 3$-th row. Similarly, the column permutation is $\pi_{\text{col}} = (0,3,1,4,2,5)$. Thus, $\bfM^\pi$ becomes 
	\begin{align*}
	\bfM^{\pi} = \begin{bmatrix}
	1 & 1 & 1 & & &\\
	1 & \omega_q & \omega_q^2 & & & \\
	& & & 1 & 1 & 1\\
	& & & 1 & \omega_q^{-1} & \omega_q^{-2}
	\end{bmatrix}.
	\end{align*}
\end{example}

\begin{claim}
	\label{claim:simplified_proof_kron_inv}
\begin{itemize}
\item[(i)] Let $a_0(z) = \sum_{j=0}^{\ell_a - 1} a_{j0} z^j$, $a_1(z) = \sum_{j=0}^{\ell_a - 1} a_{j1} z^{-j}$ and $b_{0}(z) = \sum_{j=0}^{\ell_b - 1} b_{j0} z^{j\ell_a}$, $b_{1}(z) = \sum_{j=0}^{\ell_b - 1} b_{j1} z^{-j\ell_a}$. Then, $a_{k_1}(z) b_{k_2}(z)$ for $k_1, k_2 = 0,1$ are polynomials that can be recovered from $\ell_a \ell_b$ distinct evaluation points in $\mathbb{C}$.
	
	Let $\bfD(z^j) = \text{diag}([z^j~ z^{-j}])$ and let
	\begin{align*}
	\bfX(z) &=
	\begin{bmatrix}
	\bfI_2\\
	\bfD(z)\\
	\vdots\\
	\bfD(z^{\ell_a -1})\\
	\end{bmatrix}
	\otimes
	\begin{bmatrix}
	\bfI_2\\
	\bfD(z^{\ell_a})\\
	\vdots\\
	\bfD(z^{\ell_a(\ell_b -1)})\\
	\end{bmatrix}.
	\end{align*}
	Then, if $z_i$'s are distinct points in $\mathbb{C}$, the matrix
	\begin{align*}
	[\bfX(z_1) | \bfX(z_2) | \dots | \bfX(z_{\ell_a \ell_b})],
	\end{align*}
	is nonsingular.

\item[(ii)] The matrix $[ \bfX_{i_0} | \bfX_{i_1} | \dots | \bfX_{i_{\tau-1}}]$ (defined in the proof of Theorem \ref{thm:matMatRotEmbed}) is permutation equivalent to a block-diagonal matrix with four blocks each of size $\tau \times \tau$. Each of these blocks is a Vandermonde matrix with parameters from the set $\{1, \omega_q, \omega_q^2, \dots, \omega_q^{q-1}\}$.

\end{itemize}
\end{claim}
\begin{proof}
	First we show that $a_{k_1}(z) b_{k_2}(z)$ for $k_1, k_2 = 0,1$ are polynomials that can be recovered from $\ell_a \ell_b$ distinct evaluation points in $\mathbb{C}$. Towards this end, these four polynomials can be written as
	\begin{align*}
	a_{0}(z)b_{0}(z)&=\sum_{i=0}^{\ell_a-1}\sum_{j=0}^{\ell_b-1} a_{i0}b_{j0}z^{i+j\ell_a},\\
	a_{0}(z)b_{1}(z)&=\sum_{i=0}^{\ell_a-1}\sum_{j=0}^{\ell_b-1} a_{i0}b_{j1}z^{i-j\ell_a},\\
	a_{1}(z)b_{0}(z)&=\sum_{i=0}^{\ell_a-1}\sum_{j=0}^{\ell_b-1} a_{i1}b_{j0}z^{-i+j\ell_a}, \text{~and}\\
	a_{1}(z)b_{1}(z)&=\sum_{i=0}^{\ell_a-1}\sum_{j=0}^{\ell_b-1} a_{i1}b_{j1}z^{-i-j\ell_a}.
	\end{align*}
	Upon inspection, it can be seen that each of the polynomials above has $\ell_a\ell_b$ consecutive powers of $z$. Therefore, each of these can be interpolated from $\ell_a\ell_b$ non-zero distinct evaluation points in $\mathbb{C}$.
	
	The second part of the claim follows from the above discussion. To see this we note that
	\begin{align*}
	[a_0(z) ~a_1(z)] &= [a_{00} ~a_{01}~a_{10}~a_{11}~ \dots a_{(\ell_a-1)0}~ a_{(\ell_a-1)1}] \begin{bmatrix}
	\bfI_2\\
	\bfD(z)\\
	\vdots\\
	\bfD(z^{\ell_a -1})\\
	\end{bmatrix} \text{~and}\\
	[b_0(z) ~b_1(z)] &= [b_{00} ~b_{01}~b_{10}~b_{11}~ \dots b_{(\ell_b-1)0} ~b_{(\ell_b-1)1}] \begin{bmatrix}
	\bfI_2\\
	\bfD(z^{\ell_a})\\
	\vdots\\
	\bfD(z^{\ell_a(\ell_b -1)})\\
	\end{bmatrix}.
	\end{align*}
	Furthermore, the four product polynomials under consideration can be expressed as
	\begin{align*}
	&[a_0(z) ~a_1(z)] \otimes [b_0(z) ~b_1(z)]\\
	&= \left( [a_{00} ~a_{01}~a_{10}~a_{11}~ \dots a_{(\ell_a-1)0}~ a_{(\ell_a-1)1}] \otimes [b_{00} ~b_{01}~b_{10}~b_{11}~ \dots b_{(\ell_b-1)0} ~b_{(\ell_b-1)1}] \right) \bfX(z).
	\end{align*}
	We have previously shown that all polynomials in $[a_0(z) ~a_1(z)] \otimes [b_0(z) ~b_1(z)]$ can be interpolated by obtaining their values on $\ell_a\ell_b$ non-zero distinct evaluation points. This implies that we can equivalently obtain $$\left( [a_{00} ~a_{01}~a_{10}~a_{11}~ \dots a_{(\ell_a-1)0}~ a_{(\ell_a-1)1}] \otimes [b_{00} ~b_{01}~b_{10}~b_{11}~ \dots b_{(\ell_b-1)0} ~b_{(\ell_b-1)1}] \right)$$ which means that $[\bfX(z_1) | \bfX(z_2) | \dots | \bfX(z_{\ell_a \ell_b})]$ is non-singular. This proves the statement in part (i).

The proof of the statement in (ii) is essentially an exercise in showing the permutation equivalence of several matrices by using Claim \ref{claim:diag_block_matrix} and the permutation equivalence properties of Kronecker products. For convenience, we define
\begin{align*}
\bfX_{l,A} &= \begin{bmatrix}
	\bfI \\
	\Lambda^{l}\\
	\vdots\\
	\Lambda^{l (k_A -1)}
	\end{bmatrix}, \text{~and}\\
\bfX_{l,B} &= 	\begin{bmatrix}
	\bfI \\
	\Lambda^{l k_A}\\
	\vdots\\
	\Lambda^{l k_A(k_B -1)}
	\end{bmatrix}
\end{align*}
so that $\bfX_l = \bfX_{l,A}\otimes \bfX_{l,B}$. Recall that we are analyzing the matrix $\bfX = [\bfX_{i_0} | \bfX_{i_1} | \dots | \bfX_{i_{\tau-1}}]$. An  application of Claim \ref{claim:diag_block_matrix} shows that (blank entries in the matrices below indicate zero blocks)
	\begin{align*}
\bfX_{l,A}\asymp \bfX_{l,A}^P = \begin{bmatrix}
\bfV_{l,A, 1} & \\
 &\bfV_{l,A, 2}
\end{bmatrix},~\text{and}~
\bfX_{l,B}\asymp  \bfX_{l,B}^P=\begin{bmatrix}
\bfV_{l,B, 1} &\\
		&\bfV_{l,B, 2}
	\end{bmatrix},
	\end{align*}
	where $\bfV_{l,A, 1} = [1, \omega_q^l,\cdots, \omega_q^{l(k_A-1)}]^T$, $\bfV_{l,A, 2} = [1, \omega_q^{-l},\cdots, \omega_q^{-l(k_A-1)}]^T$, $\bfV_{l,B, 1} = [1, \omega_q^{lk_A},\cdots, \omega_q^{lk_A(k_B-1)}]^T$, $\bfV_{l,B, 2} = [1, \omega_q^{-lk_A},\cdots, \omega_q^{-lk_A(k_B-1)}]^T$. Then we conclude that
	$\bfX \asymp \bfX^{P} = [\bfX_{i_0}^P|\bfX_{i_1}^P|\cdots | \bfX_{i_{\tau-1}}^P]$, where $\bfX_{l}^P=\bfX_{l,A}^P\otimes \bfX_{l,B}^P$. Next we show that
	\begin{align*}
	\bfX_l^P=\bfX_{l,A}^P\otimes \bfX_{l,B}^P\asymp \bfX_l^{P,\pi} = \begin{bmatrix}
	\bfV_{l,A,1}\otimes \bfV_{l,B,1} & & &\\
	& \bfV_{l,A,2}\otimes \bfV_{l,B,1} & &\\
	& & \bfV_{l, A,1}\otimes \bfV_{l,B,2} &\\
	& & & \bfV_{l, A,2}\otimes \bfV_{l,B,2}
	\end{bmatrix}.
	\end{align*}	
	By the definition of Kronecker product, we have
	\begin{align*}
	\bfX_{l,A}^P\otimes \bfX_{l,B}^P = \begin{bmatrix}
	\bfV_{l,A, 1} \otimes \bfX_{l,B}^P & \\
	& \bfV_{l,A, 2} \otimes \bfX_{l,B}^P
	\end{bmatrix}.
	\end{align*}
	Note that $\bfV_{l,A, i} \otimes \bfV_{l, B,j} \asymp \bfV_{l, B,j} \otimes \bfV_{l,A, i}$, then
	\begin{align*}
	& \bfV_{l, A,i}\otimes \bfX_{l,B}^P \\
	=& \bfV_{l, A,i}\otimes \begin{bmatrix}
	\bfV_{l,B,1} &\\
	& \bfV_{l, B,2}
	\end{bmatrix}\\
	\asymp & \begin{bmatrix}
	\bfV_{l,B,1} &\\
	& \bfV_{l, B,2}
	\end{bmatrix} \otimes \bfV_{l, A,i}\\
	= & \begin{bmatrix}
	\bfV_{l,B,1}\otimes \bfV_{l, A,i} & \\
	& \bfV_{l, B,2}\otimes \bfV_{l, A,i}
	\end{bmatrix}\\
	\asymp &\begin{bmatrix}
	\bfV_{l, A,i}\otimes \bfV_{l,B,1}  & \\
	& \bfV_{l, A,i}\otimes \bfV_{l, B,2}
	\end{bmatrix}.
	\end{align*}
	Thus, we can conclude that $\bfX_{l}^{P}\asymp \bfX_{l}^{P,\pi}$. In addition, we have
	\begin{align*}
	\bfV_{l,A,1}\otimes \bfV_{l,B,1} &= [1, \omega_q^{l},\cdots,\omega_q^{l(k_Ak_B-2)}, \omega_q^{l(k_Ak_B-1)}]^T,\\
	\bfV_{l,A,2}\otimes \bfV_{l,B,1} &= [\omega_q^{-l(k_A -1)}, \omega_q^{-l(k_A -2)}, \cdots, \omega_q^{-l}, 1, \omega_q^{l},\cdots, \omega_q^{l(k_A(k_B-1)-1)}, \omega_q^{lk_A(k_B-1)}]^T,\\
	\bfV_{l,A,1}\otimes \bfV_{l,B,2} &= [\omega_q^{-lk_A(k_B-1)}, \omega_q^{-l(k_A(k_B-1)-1)}, \cdots, \omega_q^{-l}, 1, \omega_q^{l}, \cdots,\omega_q^{l(k_A-2)},  \omega_q^{l(k_A-1)}]^T, \text{~and}\\
	\bfV_{l,A,2}\otimes \bfV_{l,B,2} &= [\omega_q^{-l(k_Ak_B-1)}, \omega_q^{-l(k_Ak_B-2)}, \cdots, \omega_q^{-l}, 1]^T.
	\end{align*}
	Finally applying Claim \ref{claim:diag_block_matrix} again we obtain the required result.
\end{proof}

\begin{claim}
\label{claim:threshold_gen_new}
Let $\tau_{\text{diff}} = 2k_Ak_Bp-2(k_Ak_B+pk_A+pk_B)+k_A+k_B+2p$ where $k_A, k_B$ and $p$ are positive integers with $p>1$. Then, $\tau_{\text{diff}}<0$ only if $k_A=1$ or $k_B=1$.
\end{claim}
\begin{proof}
If $k_A = 1$, then $\tau_{\text{diff}}= 1 - k_B < 0$ when $k_B > 1$; a similar argument holds when $k_B = 1, k_A > 1$. On the other hand when $k_A> 1$ and $k_B > 1$, suppose that
\begin{align}
2k_Ak_Bp+k_A+k_B+2p &< 2(k_Ak_B+pk_A+pk_B),\nonumber \\
\implies 2 + \frac{1}{k_B p} + \frac{1}{k_A p} + \frac{2}{k_A k_B} &< 2 \left( \frac{1}{p} + \frac{1}{k_B} + \frac{1}{k_A}\right) \text{~upon dividing by $k_A k_B p$}). \label{eq:homogenized}
\end{align}
We note that if $k_A, k_B$ and $p$ are all $\geq 3$, then we have a contradiction since the RHS is $\leq 2$, whereas the LHS is $> 2$. Thus, we only need to consider a limited number of cases where some of the values equal 2. These can be verified on a case by case basis.
\end{proof}

\end{document}